\documentclass[journal,draftcls,onecolumn,12pt,twoside]{IEEEtran}
\usepackage{subcaption}
\usepackage[latin9]{inputenc}
\usepackage{gensymb}
\usepackage{amsmath}
\usepackage{amssymb}
\usepackage{graphicx}
\usepackage{lipsum}
\usepackage{savesym}
\usepackage{slashbox}
\usepackage{amsthm}
\usepackage{cuted}
\usepackage{nccmath}
\usepackage{amsfonts}
\usepackage{epstopdf}
\usepackage{lipsum}
\usepackage{float}
\usepackage{color}
\usepackage{lettrine}
\usepackage{lscape}
\usepackage{multirow}
\usepackage{afterpage}
\usepackage{mathtools}
\normalsize
\definecolor{mygreen}{RGB}{28,172,0}
\definecolor{mylilas}{RGB}{170,55,241}
\def\BibTeX{{\rm B\kern-.05em{\sc i\kern-.025em b}\kern-.08em
		T\kern-.1667em\lower.7ex\hbox{E}\kern-.125emX}}

\newtheorem{corollary}{Corollary}

\newtheorem{proposition}{Proposition}
\newtheorem{remark}{Remark}
\newtheorem{keywords}{Keywords}

\def\@degree{\@latex@error{No \noexpand\degree given}\@ehc}

\begin{document}
	\title{Performance Analysis of I/Q Imbalance with Hardware Impairments over
		Fox's H-Fading Channels}
	\author{Yassine Mouchtak and Faissal El Bouanani, \textit{Senior Member, IEEE} \thanks{%
			Y. Mouchtak and F. El Bouanani are with ENSIAS, Mohammed V University in
			Rabat, Morocco,
			\textcolor{black}{e-mails:
				\{yassine.mouchtak,f.elbouanani\}@um5s.net.ma.}}}
	\maketitle
	\vspace{-2cm}
	\begin{abstract}
		Impairments baseband model in-phase and quadrature-phase Imbalance (IQI) and
		Residual hardware impairments (RHI) are two key factors degrading the
		performance of wireless communication systems (WCSs), particularly when
		high-frequency bands are employed, as in 5G systems and beyond. The impact
		of either IQI or RHI on the performance of various WCSs have been
		investigated exclusively in a separate way. To fill this gap, in this paper,
		the joint effect of both IQI and RHI on the performance of a WCS subject to sum of
		Fox's H-function fading model (SFHF)\footnote{SFHF denotes the fading
			channels whose \textcolor{black}{PDF} can be expressed as a summation of Fox's
			H-functions.} is investigated. Such a fading model generalizes most, if not
		all, of well-known fading and turbulence models. To this end, closed-form and
		asymptotic expressions for the outage probability (OP),
		\textcolor{black}{channel capacity (CC) under constant power with
			optimum rate adaptation (ORA) policy}, and average symbol error probability
		(ASEP) for both coherent and non-coherent modulation schemes. Specifically,
		all the analytical expressions are derived for three different
		scenarios: (i) ideal Tx and Rx impaired, (ii) Tx impaired and ideal Rx, and
		(iii) both Tx and Rx are impaired. Further, asymptotic expressions for OP, %
		\textcolor{black}{CC under ORA policy}, and ASEP are obtained, based on
		which, insightful discussions on the IQI and RHI impacts are made. $%
		\alpha-\mu $ and M\'alaga $\mathcal{M}$ turbulence with pointing error
		distribution models have been considered particular SFHF distribution cases. The analytical derivations, revealed by simulation results,
		demonstrate that the RF impairments' effects should be seriously taken into
		account in the design of next-generation wireless technologies.
	\end{abstract}
	
	\begin{keywords}
		{\small Channel capacity, Fox's H-fading, in-phase and quadrature-phase
			imbalance, M\'alaga $\mathcal{M}$ turbulence, $\alpha -\mu $ fading, Gray
			coded differential quadrature phase-shift keying, outage probability,
			pointing error, residual hardware impairments, symbol error rate, Average
			symbol error rate.}
	\end{keywords}
	
	\section{Introduction}
	
	The increase of mobile data traffic and the internet's demand impose a high
	spectral efficiency, low latency, and massive connectivity requirements
	towards fifth-generation (5G) wireless networks and beyond. In this context,
	direct conversion transceivers employing quadrature up/down conversion of
	radio-frequency (RF) signals have attracted considerable attention owing to
	their potential in reducing power consumption and cost. In addition, they
	request neither external intermediate frequency filters nor image rejection
	filters \cite{trans1} and \cite{trans2}. Nevertheless, such transceivers
	suffer from inevitable RF front-end imperfections because of component
	mismatches and manufacturing defects that degrade the performance of
	wireless communication systems (WCS) \cite{imperf1}. Furthermore, the
	amplitude and/or phase mismatch between the in-phase (I) and
	quadrature-phase (Q) at the transmitter (Tx) or/and the receiver (Rx)
	branches, so-called I/Q imbalance (IQI) and known also as zero intermediate
	frequency, is one of the major performance-limiting impairments causing
	self-interference and degrading the reliability, particularly in high-rate
	WCS \cite{imperf2}. To this end, there have been several works modeling IQI,
	studying its impact on system performance, and tackling the problem. For
	instance, in \cite{work1}-\cite{workfin} and the citations referred therein,
	the IQI influence on the system performance has been extensively assessed.
	Moreover, in \cite{algo1}-\cite{algo2} and \cite{calib1}-\cite{calib3}, as
	well as the references therein, various algorithms for IQI parameters'
	estimation and calibration have been proposed, respectively. Even though,
	such IQI mitigating mechanisms fail to completely eliminate hardware
	impairments. As a consequence, residual hardware impairments (RHI), which
	are added to the transmitted or/and the received signals, degrade further
	the performance of WCS \cite{RHI1}-\cite{RHIfin}. On the other hand, several
	fading models have been proposed over years to model either the fading or
	the joint shadowing/fading phenomena. Therefore, dealing with a generalized
	distribution modeling such behaviors is of paramount importance in providing
	unified expressions for WCS performance criteria. Towards this end, the
	Fox's H-function (FHF) random variable (RV), has been proposed as a unified
	distribution subsuming most fading environment models. Precisely, the
	probability density function (PDF) of various models can be expressed in
	terms of a (i) single FHF (e.g., $\alpha$-$\mu$, Gamma-Gamma (GG),
	Generalized Bessel-K, Fisher-Snedecor $\mathcal{F}$), (ii) finite sum of
	FHFs (e.g., M\'alaga $\mathcal{M}$), and (iii) infinite summation of FHFs
	(e.g., shadowed $\kappa$-$\mu$) \cite{AccessOurs}. Moreover, FHF has various
	powerful properties, namely the (i) product, quotient, and power of FHF RVs
	are FHF RV \cite{HFOX}, (ii) sum of FHF RVs can be tightly approximated by
	an FHF \cite{WCL}-\cite{IRS}, and (iii) numerous complex transforms of FHFs,
	e.g. Laplace and Mellin transforms, can be written as FHFs \cite{kilbas}. By
	leveraging such properties, the performance analysis, and physical layer
	security of intelligent reflecting surface aided cooperative WCS with
	spatial diversity over FHF fading channels in presence of co-channel
	interference (CCI), jamming signals, or/and eavesdroppers can be
	straightforwardly investigated \cite{IRS}, \cite{Hwork3}.\newline
	On the other hand, the instantaneous symbol error probability (SEP) is a
	crucial performance metric quantifying the instantaneous communication
	reliability. Importantly, such a metric is often provided in complicated
	integral form depending on the employed modulation technique \cite{gqfunc}-%
	\cite{talambura}. Therefore, the closed-form of the average SEP (ASEP)
	remains a big challenge for numerous communication systems due to the
	mathematical intractability of both the end-to-end (e2e) fading model and
	the SEP expression. Obviously, obtaining accurate and simpler bounds or
	approximate expressions for the ASEP is strongly depending on the SEP's
	ones. To this end, it has been shown in \cite{yassineaccess} that the use of
	the Trapezoidal rule technique leads to a simple and tight approximate
	expression for the SEP's integral form. In fact, the mathematical
	tractability of such an approximation lies in its exponential form, making
	from the ASEP approximate expression a summation of Laplace transforms of
	the SNR's PDF.
	
	\subsection{Related work}
	
	Recently, the impact of RF impairments on the performance of a WCS has been
	the focus of various works. In \cite{weib}, assuming IQI at the receiver
	only, the effects of the IQI on a generalized frequency division
	multiplexing system under Weibull fading channels was investigated by
	evaluating the ASER for $M$-ary quadrature amplitude modulation ($M$-QAM),
	while in \cite{iqidqpsk}, the IQI effects on the bit error rate of Gray
	coded differential quadrature phase-shift keying (GC-DQPSK) over Rayleigh
	fading channel ware evaluated. The IQI impact on the outage probability of
	both single- and multi- carrier systems over cascaded Nakagami-$m$ fading
	channels have been quantified in \cite{multicarier}, whereas in \cite{gama},
	a free-space optical system using subcarrier quadrature phase-shift keying
	(QPSK) is assessed for Gamma-Gamma fading channels with IQI at the receiver.
	On the other hand, the performance analysis of systems suffering from RHI
	has attracted considerable attention \cite{RHI1}-\cite{RHIfin}.
	Specifically, tight closed-form and asymptotic outage probability (OP)
	expressions for an impaired RHI cognitive amplify-and-forward (AF)
	multi-relay cooperative network undergoing Rayleigh fading model ware
	derived in \cite{RHI1}. Likewise, OP and \textcolor{black}{CC} for an impaired
	RHI WCS with CCI and imperfect channel state information (CSI) were
	investigated in \cite{RHI2}. Very recently, OP of an impaired RHI hybrid
	satellite-terrestrial relay network was analyzed in \cite{RHI3}, while the
	performance of non-orthogonal multiple access-based AF relaying network
	subject to RHI and Nakagami-$m$ fading model was investigated in \cite%
	{RHIfin}.
	
	\subsection{Motivation and contributions}
	
	Capitalizing on the above, the RF impairments have a paramount impact on the
	performance of wireless communication systems. Although the performance of
	WCSs, under the assumption of perfect CSI at the receiver, is treated in
	most of the existing works. Nonetheless, such a hypothesis is not practical
	as the receiver requires often complex channel estimation algorithms
	providing, in some cases, uncertain values, which would result in
	performance degradation. To overcome this limitation, differential
	modulation, among other solutions, can be employed particularly for
	low-power wireless systems \cite{abouie}. In \cite{owc}, it is shown that
	employing GC-DQPSK in the inter-satellite optical wireless system enhances
	the performance compared to other modulation techniques. Furthermore, such a
	scheme simplifies the detection as the channel estimation and tracking
	becomes needless, which reduces the cost and complexity of the receiver.
	
	To the best of the authors' knowledge, the joint effect of both IQI and RHI
	on the performance of a WCS undergoing FHF fading channel has not been
	addressed in the open technical literature. Motivated by this, in this
	paper, we analyze the performance of an impaired IQI/RHI WCS experiencing
	FHF fading model. Specifically, by deriving the cumulative distribution
	function (CDF) expression of the signal-to-interference-plus-noise-ratio
	(SINR), closed-form and asymptotic expressions for the OP,
	\textcolor{black}{channel capacity (CC) under constant power with optimum rate
		adaptation (ORA) policy}, and ASEP for different modulation schemes are
	presented for the proposed impaired IQI/RHI WCS. Pointedly, the key
	contributions of this paper can be summarized as follows:
	
	\begin{itemize}
		\item Expressions of the SINR per symbol at the receiver's input impaired by
		both IQI and RHI are presented for different cases: (i) Tx Impaired by both
		IQI and RHI, (ii) Rx Impaired by both IQI and RHI, and (iii) both Tx and Rx
		are impaired.
		
		\item Based on the statistical properties of the SINR, novel closed-form and
		asymptotic expressions for the OP as well as for CC are derived.
		
		\item A simple exponential approximation for the SEP for various coherent
		and non-coherent used modulation techniques is proposed, based on which,
		closed-form and asymptotic expressions for the ASEP are derived.
		
		\item Finally, to gain further insights into the system's performance, the
		derived expressions are evaluated for both $\alpha-\mu$ fading and M\'alaga $%
		\mathcal{M}$ turbulence channels with pointing errors (TCPE).
	\end{itemize}
	
	\subsection{Organization}
	
	The rest of this paper is structured as follows. Section II is dedicated to
	the statistical properties of the considered WCS by retrieving the PDF and
	CDF of SINR, while closed-form and asymptotic expressions for the OP, CC,
	and ASEP are derived in Section III. All the analytical results are
	retrieved for two particular cases of FHF RV, namely $\alpha-\mu$ and
	M\'alaga $\mathcal{M}$ distributions in Section IV. In Section V, Numerical
	results and insightful discussions are presented. Finally, some conclusions are drawn in Section VI.
	
	\subsection{Notations}
	
	For the sake of clarity, the notations used throughout the paper are
	summarized in Table \ref{notations}.
	\begin{table}[tbh]
		\captionsetup{font=small}
		\caption{Symbols and notations.}
		\label{notations}\centering
		\begin{tabular}{p{1.5cm}|p{6cm}|p{3.5cm}|p{5.5cm}}
			\hline
			\centering \textit{Symbol} & \centering \textit{Meaning} & \centering \textit{Symbol} & \qquad \qquad
			\qquad \qquad \textit{Meaning} \\ \hline\hline
			\centering$f_{X}(\cdot )$ & Probability density function of the random
			variable $X$ & \centering$\gamma _{id}$ & SNR for ideal RF end-to-end \\
			\hline
			\centering$H_{p,q}^{m,n}(\cdot )$ & Univariate Fox's H-function & \centering$%
			\gamma _{\varkappa }$ & SINR of impaired RF end-to-end \\ \hline
			\centering$L^{\left( \ell \right) }$ & Complex contour ensuring the
			convergence of the $\ell $th Mellin-Barnes integral & \centering$F_{X}(\cdot
			)$ & Cumulative distribution function of a random variable $X$ \\ \hline
			\centering$s$ & Ideal transmitter signal & \centering$P_{\text{out}}(\cdot )$
			& Outage probability $\left( \text{OP}\right) $ \\ \hline
			\centering$\varkappa $ & Transmitter (t) or receiver (r) & \centering$%
			C^{\varkappa }$ & Channel capacity (CC) \\ \hline
			\centering$s_{\varkappa }$ & Transmitted or received signal & \centering$%
			H_{p_{1},q_{1}:p_{2},q_{2}:p_{3},q_{3}}^{m_{1},n_{1}:m_{2},n_{2}:m_{3},n_{3}}(\cdot )
			$ & Bivariate Fox's H-function \\ \hline
			\centering$G_{\varkappa }$ & IQI gain & \centering$H\left( \gamma \right) $
			& Symbol error probability $\left( \text{SEP}\right) $ \\ \hline
			\centering IRR$_{\varkappa }$ & Image rejection ratio & \centering$P_{s}$ &
			Average symbol error probability $\left( \text{ASEP}\right) $ \\ \hline\hline
		\end{tabular}%
	\end{table}
	
	\section{System and Signal model}
	
	Consider a WCS, which consists of a single transmitter and receiver and
	subject to IQI and RHI impairments. In our setup, both Tx and Rx are
	equipped with a single antenna and the flat fading channel is modeled by a
	generalized SFHF RV. The SNR's PDF can be expressed as
	\begin{equation}
	f_{\gamma }\left( \gamma \right) =\sum_{\ell =1}^{L}\psi _{\ell }{\small H}%
	_{p{\small ,q}}^{m{\small ,n}}\left( \phi _{\ell }\gamma \left\vert
	\begin{array}{c}
	\left( a_{i}^{(\ell )},A_{i}^{(\ell )}\right) _{i=1:p} \\
	\left( b_{i}^{(\ell )},B_{i}^{(\ell )}\right) _{i=1:q}%
	\end{array}%
	\right. \right) ,  \label{pdf}
	\end{equation}%
	where $\psi _{\ell }$ and $\phi _{\ell }$ are any two positive constants
	satisfying $\int_{0}^{\infty }f_{\gamma }\left( \gamma \right) d\gamma =1$, $%
	\text{i.e., }$%
	\begin{equation}
	\sum_{\ell =1}^{L}E_{\ell }\frac{\psi _{\ell }}{\phi _{\ell }}=1,
	\label{Condition}
	\end{equation}%
	with
	\begin{gather}
	E_{\ell }=\frac{\prod\limits_{i=1}^{m}\Gamma \left( \mathcal{B}_{i}^{(\ell
			)}\right) \prod\limits_{i=1}^{n}\Gamma \left( 1-\mathcal{A}_{i}^{(\ell
			)}\right) }{\prod\limits_{i=n+1}^{p}\Gamma \left( \mathcal{A}_{i}^{(\ell
			)}\right) \prod\limits_{i=m+1}^{q}\Gamma \left( 1-\mathcal{B}_{i}^{(\ell
			)}\right) },\ell =1..L,  \label{Ell} \\
	\mathcal{A}_{i}^{(\ell )}=a_{i}^{(\ell )}+A_{i}^{(\ell )},  \label{mathA} \\
	\mathcal{B}_{i}^{(\ell )}=b_{i}^{(\ell )}+B_{i}^{(\ell )},  \label{mathB}
	\end{gather}%
	and
	\begin{equation}
	{\small H}_{p{\small ,q}}^{m{\small ,n}}\left( \phi _{\ell }x\left\vert
	\begin{array}{c}
	\left( a_{i}^{(\ell )},A_{i}^{(\ell )}\right) _{i=1:p} \\
	\left( b_{i}^{(\ell )},B_{i}^{(\ell )}\right) _{i=1:q}%
	\end{array}%
	\right. \right) =\frac{1}{2\pi j}\int_{\mathcal{L}_{v}^{\left( _{\ell
			}\right) }}\frac{\prod\limits_{i=1}^{m}\Gamma \left( b_{i}^{(\ell
			)}+B_{i}^{(\ell )}v\right) \prod\limits_{i=1}^{n}\Gamma \left(
		1-a_{i}^{(\ell )}-A_{i}^{(\ell )}v\right) }{\prod\limits_{i=n+1}^{p}\Gamma
		\left( a_{i}^{(\ell )}+A_{i}^{(\ell )}v\right)
		\prod\limits_{i=m+1}^{q}\Gamma \left( 1-b_{i}^{(\ell )}-B_{i}^{(\ell
			)}v\right) }x^{-v}dv.  \label{Hffox}
	\end{equation}%
	denotes the univariate FHF, $j=\sqrt{-1}$, $\mathcal{L}_{v}^{\left( _{\ell
		}\right) }$ is an appropriately chosen complex contour, ensuring the
	convergence of the $\ell $th above Mellin-Barnes integral, $\Gamma \left(
	.\right) $ denotes the Euler Gamma function \cite{integraltable}, $%
	a_{i}^{(\ell )}$ and $b_{i}^{(\ell )}$ are real-valued numbers, whereas $%
	A_{i}^{(\ell )}$ and $B_{i}^{(\ell )}$ are two real positive numbers.
	
	\subsection{RF impairments baseband model}
	
	\subsubsection{IQI impairment model}
	
	The time-domain baseband representation of the IQI-impaired signal, at
	either the transmitter or the receiver, can be expressed as \cite{imperf1}
	\begin{equation}
	s_{\varkappa ,IQI}=G_{\varkappa ,1}s+G_{\varkappa ,2}s^{\ast },
	\end{equation}%
	where $\varkappa $ denotes $t$ or $r$ for the transmitter (Tx) or the
	receiver (Rx) node, respectively, $s$ represents the ideal baseband signal
	under perfect in-phase and quadrature-phase (I/Q) matching. Further, the IQI
	parameters $G_{\varkappa ,1}$ and $G_{\varkappa ,2}$ are expressed,
	respectively, as
	\begin{equation}
	G_{\varkappa ,1}=\frac{1+g_{\varkappa }e^{\varepsilon _{\varkappa }\varphi
			_{\varkappa }}}{2},  \label{G1X}
	\end{equation}%
	\begin{equation}
	G_{\varkappa ,2}=\frac{1-g_{\varkappa }e^{-\varepsilon _{\varkappa }\varphi
			_{\varkappa }}}{2},  \label{G2X}
	\end{equation}%
	where $\varepsilon _{t}=+j,$ and $\varepsilon _{r}=-j$ with $\ j=\sqrt{-1}$,
	while $g_{\varkappa }$ and $\varphi _{\varkappa }$ represent the gain and
	phase mismatch for the node $\varkappa $, respectively. Particularly, for
	perfect I/Q matching, $g_{\varkappa }=1$ and $\varphi _{\varkappa }=0$,
	leading to $G_{\varkappa ,1}=1$ and $G_{\varkappa ,2}=0.$ On the other hand,
	the corresponding image rejection ratio $\left( \text{IRR}\right) $, which
	determines the amount of image frequency attenuation, is linked to the IQI
	parameters as follow
	\begin{equation}
	IRR_{\varkappa }=\frac{\left\vert G_{\varkappa ,1}\right\vert ^{2}}{%
		\left\vert G_{\varkappa ,2}\right\vert ^{2}}.  \label{IRRdef}
	\end{equation}%
	Substituting (\ref{G1X}) and (\ref{G2X}) into (\ref{IRRdef}), yields
	\begin{equation}
	IRR_{\varkappa }=\frac{g_{\varkappa }^{2}+2g_{\varkappa }\cos \left( \varphi
		_{\varkappa }\right) +1}{g_{\varkappa }^{2}-2g_{\varkappa }\cos \left(
		\varphi _{\varkappa }\right) +1}.  \label{IRRdef2}
	\end{equation}%
	For practical analog RF front-end electronics, the IRR's value is typically
	in the range $\left[ 30\text{ dB},40\text{ dB}\right] $ indicating a
	possible combination of 0.2 to 0.6 dB of gain mismatch and $1^{\circ}$ to $5^{\circ}$ of
	phase imbalance \cite{Rfmicro}.
	
	\subsubsection{RHI model}
	In \cite{imperf1} and \cite{studer}, it has been shown that the complex
	normal distribution fits accurately the residual distortion noise behavior,
	with zero mean and average power proportional that of the signal of
	interest, i.e.,
	\begin{equation}
	\eta_{\varkappa}\sim\mathcal{CN}\left(0,\kappa_{\varkappa}^{2}P_{\varkappa}%
	\right),
	\end{equation}
	where $\kappa_{\varkappa}^{2}$ is a proportionality parameter describing the
	residual additive impairments severity at node $\varkappa,$ while $P_{t}$
	and $P_{r}$ refer to the transmit and receive power, respectively. Note that
	$\kappa_{t}=\kappa_{r}=0$ is corresponding to ideal hardware on both sides
	case.
	
	\subsection{Received Signal Model}
	
	Under the Assumption of single antennas at both Tx and Rx, we provide in
	this subsection the expression of the received signal in various scenarios
	of Tx and/or Rx impairments.
	
	\subsubsection{Ideal scenario}
	
	The received signal with ideal Tx and Rx RF front-end is represented as
	\begin{equation}
	r_{id}=h\mathrm{s}+\mathrm{n},  \label{idealreceived signal}
	\end{equation}%
	where $h$ and $n$ denote the channel coefficient and the additive complex
	Gaussian receiver's noise, respectively. Moreover, in this case, the
	instantaneous signal-to-noise ratio (SNR) per symbol at the receiver input
	is expressed as
	\begin{equation}
	\gamma _{id}=\frac{P_{s}}{N_{0}}\left\vert h\right\vert ^{2},
	\label{idealSNR}
	\end{equation}%
	where $P_{s}$ is the energy per transmitted symbol and $N_{0}$ represents
	the single-sided power spectral density noise. In what follows, we will focus
	on deriving the expression of the signal-to-interference-plus-noise-ratio
	(SINR) in which Tx and/or Rx suffers from both IQI and RHI impairments.
	
	\subsubsection{Tx impaired by both IQI and RHI}
	
	In this scenario, it is assumed that Tx experiences IQI and RHI, while the
	Rx RF front-end is ideal (i.e., without RHI nor IQI). Thus, the transmitted
	signal can be represented as 
	\begin{eqnarray}
	\mathrm{s}_{t,imp} &=&\mathrm{s}_{t,IQI}+\eta _{t},\notag \\
	&=&G_{t,1}\mathrm{s}+G_{t,2}\mathrm{s}^{\ast }+\eta _{t}.
	\end{eqnarray}
	Consequently, the received signal can be expressed as
	\begin{eqnarray}
	y &=&h\mathrm{s}_{t,imp}+\mathrm{n},\notag \\
	&=&hG_{t,1}\mathrm{s}+hG_{t,2}\mathrm{s}^{\ast }+h\eta _{t}+\mathrm{n}.
	\end{eqnarray}
	It follows that the instantaneous SINR can be evaluated using (\ref{idealSNR}%
	), as
	\begin{equation}
	\gamma _{t}=\frac{\left\vert G_{t,1}\right\vert ^{2}}{\left\vert
		G_{t,2}\right\vert ^{2}+\kappa _{t}^{2}+\frac{1}{\gamma _{id}}}.
	\label{snrTxonly}
	\end{equation}
	
	\subsubsection{Rx impaired by both IQI and RHI}
	
	In a similar manner, the Rx experiences here both IQI and RHI, while the Tx
	RF front-end is supposed ideal. To this end, the received signal is given as
	\begin{eqnarray}
	y &=&G_{r,1}\left( h\mathrm{s}+\mathrm{n}\right) +G_{r,2}\left( h^{\ast }\mathrm{s}^{\ast }+\mathrm{n}^{\ast
	}\right) +\eta _{r} \notag \\
	&=&G_{r,1}h\mathrm{s}+G_{r,2}h^{\ast }\mathrm{s}^{\ast }+G_{r,1}\mathrm{n}+G_{r,2}\mathrm{n}^{\ast }+\mathrm{\eta} _{r}.
	\label{yRXonly}
	\end{eqnarray}
	Thus, the corresponding instantaneous SINR can be simplified using (\ref%
	{idealSNR}) to
	\begin{equation}
	\gamma _{r}=\frac{\left\vert G_{r,1}\right\vert ^{2}}{\left\vert
		G_{r,2}\right\vert ^{2}+\kappa _{r}^{2}+\frac{\Lambda }{\gamma _{id}}},
	\label{snrRxonly}
	\end{equation}
	with%
	\begin{equation}
	\Lambda =\left\vert G_{r,1}\right\vert ^{2}+\left\vert G_{r,2}\right\vert
	^{2}.
	\end{equation}
	\subsubsection{Both Tx and Rx are impaired}
	
	In this case, both Tx and Rx experience IQI and RHI. As a result, the
	equivalent baseband received signal can be written as
	\begin{eqnarray}
	y &=&G_{r,1}\left[ h\left( G_{t,1}\mathrm{s}+G_{t,2}\mathrm{s}^{\ast }+\eta _{t}\right)+\mathrm{n}%
	\right]+G_{r,2}\left[ h\left( G_{t,1}\mathrm{s}+G_{t,2}\mathrm{s}^{\ast }+\eta _{t}\right)+\mathrm{n}\right]
	^{\ast }+\eta _{r}\notag \\
	&=&\left[ \zeta _{11}h+\zeta _{22}h^{\ast }\right]\mathrm{s}+\left[ \zeta
	_{12}h+\zeta _{21}h^{\ast }\right]\mathrm{s}^{\ast }+G_{r,1}h\eta _{t}+G_{r,2}h^{\ast }\eta _{t}^{\ast }+G_{r,1}\mathrm{n}+G_{r,2}\mathrm{n}^{\ast }+\eta _{r},
	\label{yTXRX}
	\end{eqnarray}
	where $\zeta _{ij}$ are the elements of the following matrix%
	\begin{equation}
	\mathbf{\zeta }=\left(
	\begin{array}{cc}
	G_{r,1}G_{t,1} & G_{r,1}G_{t,2} \\
	G_{r,2}G_{t,1}^{\ast } & G_{r,2}G_{t,2}^{\ast }%
	\end{array}%
	\right) .  \label{matrix}
	\end{equation}%
	The corresponding instantaneous SINR is then expressed as
	\begin{equation}
	\gamma _{t/r}=\frac{\left\vert \zeta _{11}h+\zeta _{22}h^{\ast }\right\vert
		^{2}P_{s}}{\left\vert \zeta _{12}h+\zeta _{21}h^{\ast }\right\vert
		^{2}P_{s}+\Lambda \left\vert h\right\vert ^{2}\kappa _{t}^{2}P_{s}+\Lambda
		\sigma ^{2}+\kappa _{r}^{2}P_{r}}.  \label{snrTXRX}
	\end{equation}%
	On the other hand, one can check from $\left( \ref{IRRdef}\right) $
	and $\left( \ref{matrix}\right) $ that
	\begin{equation}
	\left\vert \zeta _{11}\right\vert \left\vert \zeta _{22}\right\vert
	=\left\vert \zeta _{12}\right\vert \left\vert \zeta _{21}\right\vert =\frac{%
		\left\vert G_{r,1}\right\vert ^{2}\left\vert G_{t,1}\right\vert ^{2}}{\sqrt{%
			IRR_{r}IRR_{t}}}.
	\end{equation}%
	As mentioned above, IRR$_{\varkappa }$ practically lies in the range $\left[
	30\text{ dB}-40\text{ dB}\right] $ and the gains $\left\vert G_{\varkappa
		,1}\right\vert \leq 1,$ it follows that $\left\vert \zeta _{11}\right\vert
	\left\vert \zeta _{22}\right\vert $ as well as $\left\vert \zeta
	_{12}\right\vert \left\vert \zeta _{21}\right\vert $ are relatively smaller$.
	$ As a result, $\gamma _{t/r}$ can be accurately approximated by
	\begin{equation}
	\gamma _{t/r}\approx \frac{\left\vert \zeta _{11}\right\vert ^{2}+\left\vert
		\zeta _{22}\right\vert ^{2}}{\left\vert \zeta _{12}\right\vert
		^{2}+\left\vert \zeta _{21}\right\vert ^{2}+\Lambda \kappa _{t}^{2}+\kappa
		_{r}^{2}+\frac{\Lambda }{\gamma _{id}}}.
	\end{equation}%
	Subsequently, using $\left( \ref{snrTxonly}\right)$, $\left( \ref{snrRxonly}\right) $, and
	$\left( \ref{snrTXRX}\right) $, the instantaneous SINR can be expressed for
	all cases in a unified form as
	\begin{equation}
	\gamma _{\varkappa }=\frac{\omega _{\varkappa }}{\varrho _{\varkappa }+\frac{%
			\tau _{\varkappa }}{\gamma _{id}}},  \label{SINR}
	\end{equation}%
	where $\omega _{\varkappa }$, $\varrho _{\varkappa }$ and $\tau _{\varkappa }
	$ are summarized in Table \ref{parametresIQI}.
	\begin{table}[tbh]
		\centering
		\captionsetup{font=small}
		\caption{$\protect\omega _{\varkappa }$, $\protect\varrho _{\varkappa }$ and
			$\protect\tau _{\varkappa }$ per impairment.}$\centering%
		\begin{tabular}{c|c|c|c}
		\hline\hline
		\backslashbox{$\varkappa$}{\textit{Parameters}} & $\omega _{\varkappa }$ & $%
		\varrho _{\varkappa }$ & $\tau _{\varkappa }$ \\ \hline
		$t$ & $\left\vert G_{t,1}\right\vert ^{2}$ & $\left\vert G_{t,2}\right\vert
		^{2}+\kappa _{t}^{2}$ & $1$ \\ \hline
		$r$ & $\left\vert G_{r,1}\right\vert ^{2}$ & $\left\vert G_{r,2}\right\vert
		^{2}+\kappa _{r}^{2}$ & $\Lambda $ \\ \hline
		$t/r$ & $\left\vert \zeta _{11}\right\vert ^{2}+\left\vert \zeta
		_{22}\right\vert ^{2}$ & $\left\vert \zeta _{12}\right\vert ^{2}+\left\vert
		\zeta _{21}\right\vert ^{2}+\kappa _{r}^{2}+\Lambda \kappa _{t}^{2}$ & $%
		\Lambda $ \\ \hline\hline
		\end{tabular}%
		$%
		\label{parametresIQI}
	\end{table}
	\begin{remark}
		One can notice from (\ref{SINR}), that the SINR is upper bounded by $\frac{%
			\omega _{\varkappa }}{\varrho _{\varkappa }}$. On the other hand, for ideal
		RF front e2e, namely (i) ideal RHI $($i.e., $\kappa _{\varkappa }=0),$ and
		(ii) ideal IQI $($i.e., $G_{\varkappa ,1}=1$, $G_{\varkappa ,2}=0)$, the
		aforesaid parameters can be simplified to $\omega _{\varkappa }=1,$ $\varrho
		_{\varkappa }=0,$ and $\tau _{\varkappa }$ $=1.$ \label{remark1}
	\end{remark}
	\section{Performance Analysis}
	
	\subsection{Outage Probability}
	
	\subsubsection{Exact analysis}
	
	For a fixed impairment $\varkappa $, and SINR threshold $\gamma _{\text{th}%
	}^{(\varkappa )}$, the OP can be defined as the probability the SINR fall
	below $\gamma _{\text{th}}^{(\varkappa )}.$ Leveraging Remark \ref{remark1},
	such a metric can be defined as
	\begin{equation}
	P_{\text{out}}^{\left( \varkappa \right) }=F_{\gamma _{\varkappa }}\left(
	\gamma _{\text{th}}^{(\varkappa )}\right) ,\gamma _{\text{th}}^{(\varkappa
		)}\leq \frac{\omega _{\varkappa }}{\varrho _{\varkappa }},
	\label{Poutgeneral}
	\end{equation}%
	with
	\begin{equation}
	F_{\gamma _{\varkappa }}\left( \gamma \right) =F_{\gamma _{id}}\left( \frac{%
		\tau _{\varkappa }\gamma }{\omega _{\varkappa }-\varrho _{\varkappa }\gamma }%
	\right) ,
	\end{equation}%
	and $F_{\gamma _{id}}\left( .\right) $ denotes the CDF of ideal e2e SNR and
	which can be evaluated relying to (\ref{pdf}) as
	\begin{eqnarray}
	F_{\gamma _{id}}\left( \gamma \right)  &=&1-\int_{\gamma }^{\infty
	}f_{\gamma _{id}}\left( t\right) dt  \notag \\
	&=&1-\sum_{\ell =1}^{L}\int_{\gamma }^{\infty }\psi _{\ell }{\small H}_{p%
		{\small ,q}}^{m{\small ,n}}\left( \phi _{\ell }t\left\vert
	\begin{array}{c}
	\left( a_{i}^{\left( \ell \right) },A_{i}^{\left( \ell \right) }\right) _{i=%
		{\small 1}:p} \\
	\left( b_{i}^{\left( \ell \right) },B_{i}^{\left( \ell \right) }\right) _{i=%
		{\small 1}:q}%
	\end{array}%
	\right. \right) dt  \notag \\
	&=&1-\sum_{\ell =1}^{L}\frac{\psi _{\ell }}{\phi _{\ell }}{\small H}_{p+1%
		{\small ,q+1}}^{m+1{\small ,n}}\left( \phi _{\ell }\gamma \left\vert
	\begin{array}{c}
	\Psi ^{\left( \ell \right) },\left( 1,1\right)  \\
	\left( 0,1\right) ,\Upsilon ^{\left( \ell \right) }%
	\end{array}%
	\right. \right) ,  \label{CDFideal}
	\end{eqnarray}%
	with
	\begin{gather}
	\Psi ^{\left( \ell \right) }=\left\{ \left( \mathcal{A}_{i}^{\left( \ell
		\right) },A_{i}^{\left( \ell \right) }\right) _{i=1:n};\left( \mathcal{A}%
	_{i}^{\left( \ell \right) },A_{i}^{\left( \ell \right) }\right)
	_{i=n+1:p}\right\} , \\
	\text{and}  \notag \\
	\Upsilon ^{\left( \ell \right) }=\left\{ \left( \mathcal{B}_{i}^{\left( \ell
		\right) },B_{i}^{\left( \ell \right) }\right) _{i=1:m_{{}}};\left( \mathcal{B%
	}_{i}^{\left( \ell \right) },B_{i}^{\left( \ell \right) }\right)
	_{i=m+1:q_{{}}}\right\} ,
	\end{gather}%
	and $\mathcal{A}_{i}^{\left( \ell \right) }$, and $\mathcal{B}_{i}^{\left(
		\ell \right) }$ are defined in (\ref{mathA}) and (\ref{mathB}),
	respectively. Hence for $\gamma \geq 0$,
	\begin{equation}
	F_{\gamma _{\varkappa }}\left( \gamma \right) =1-\sum_{\ell =1}^{L}\frac{%
		\psi _{\ell }}{\phi _{\ell }}{\small H}_{p+1{\small ,q+1}}^{m+1{\small ,n}%
	}\left( \frac{\phi _{\ell }\tau _{\varkappa }\gamma }{\omega _{\varkappa
		}-\varrho _{\varkappa }\gamma }\left\vert
	\begin{array}{c}
	\Psi ^{\left( \ell \right) },\left( 1,1\right)  \\
	\left( 0,1\right) ,\Upsilon ^{\left( \ell \right) }%
	\end{array}%
	\right. \right) .  \label{CDFIQI}
	\end{equation}
	\subsubsection{Asymptotic analysis}
	
	It is worthy to mention that the scale $\phi _{\ell }$ is usually inversely
	proportional to the average SNR $\overline{\gamma }_{id}$ \cite[Tables II-V]%
	{tables}$.$ Hence, at high SNR values, $\phi _{\ell }$ tends to $0$, and
	then the FHF given in $\left( \ref{CDFIQI}\right) $, can be asymptotically
	approximated in the high SNR regime as
	\begin{equation}
	{\small H}_{p+1{\small ,q+1}}^{m+1{\small ,n}}\left( \frac{\phi _{\ell }\tau
		_{\varkappa }\gamma }{\omega _{\varkappa }-\varrho _{\varkappa }\gamma }%
	\left\vert
	\begin{array}{c}
	\Psi ^{\left( \ell \right) },\left( 1,1\right) \\
	\left( 0,1\right) ,\Upsilon ^{\left( \ell \right) }%
	\end{array}%
	\right. \right) \sim E_{\ell }-\sum_{i=1}^{m}\mathcal{F}_{\ell ,i}\left(
	\frac{\phi _{\ell }\tau _{\varkappa }\gamma }{\omega _{\varkappa }-\varrho
		_{\varkappa }\gamma }\right) ^{\frac{\mathcal{B}_{i}^{\left( \ell \right) }}{%
			B_{i}^{\left( \ell \right) }}}.  \label{ApproFHF}
	\end{equation}%
	where $E_{\ell }$ is being defined in (\ref{Ell}) and
	\begin{equation}
	\mathcal{F}_{\ell ,i}=\frac{\prod\limits_{k=1,k\neq i}^{m}\Gamma \left(
		\mathcal{B}_{k}^{\left( \ell \right) }-\mathcal{B}_{i}^{\left( \ell \right) }%
		\frac{B_{k}^{\left( \ell \right) }}{B_{i}^{\left( \ell \right) }}\right)
		\prod\limits_{k=1}^{n}\Gamma \left( 1-\mathcal{A}_{k}^{\left( \ell \right) }+%
		\mathcal{B}_{i}^{\left( \ell \right) }\frac{A_{k}^{\left( \ell \right) }}{%
			B_{i}^{\left( \ell \right) }}\right) }{\mathcal{B}_{i}^{\left( \ell \right)
		}\prod\limits_{k=n+1}^{p}\Gamma \left( \mathcal{A}_{k}^{\left( \ell \right)
		}-\mathcal{B}_{i}^{\left( \ell \right) }\frac{A_{k}^{\left( \ell \right) }}{%
			B_{i}^{\left( \ell \right) }}\right) \prod\limits_{k=m+1}^{q}\Gamma \left( 1-%
		\mathcal{B}_{k}^{\left( \ell \right) }+\mathcal{B}_{i}^{\left( \ell \right) }%
		\frac{B_{k}^{\left( \ell \right) }}{B_{i}^{\left( \ell \right) }}\right) }.
	\label{Fli}
	\end{equation}%
	, such that $\mathcal{B}_{k}^{\left( \ell \right) }-\mathcal{B}_{i}^{\left( \ell
		\right) }\frac{B_{k}^{\left( \ell \right) }}{B_{i}^{\left( \ell \right) }}$
	and $\mathcal{A}_{k}^{\left( \ell \right) }-\mathcal{B}_{i}^{\left( \ell
		\right) }\frac{A_{k}^{\left( \ell \right) }}{B_{i}^{\left( \ell \right) }}$
	are non-negative integers.
	
	Now, by substituting $\left( \ref{ApproFHF}%
	\right) $ into $\left( \ref{CDFIQI}\right) $ and using (\ref{Condition}),
	the OP can be asymptotically approximated for high SNR values as
	\begin{eqnarray}
	P_{\text{out}}^{\left( \varkappa \right) } &\sim &1-\sum_{\ell =1}^{L}\frac{%
		\psi _{\ell }}{\phi _{\ell }}\left[ E_{\ell }-\sum_{i=1}^{m}\mathcal{F}%
	_{\ell ,i}\left( \frac{\phi _{\ell }\tau _{\varkappa }\gamma _{\text{th}%
		}^{\left( \varkappa \right) }}{\omega _{\varkappa }-\varrho _{\varkappa
		}\gamma _{\text{th}}^{\left( \varkappa \right) }}\right) ^{\frac{\mathcal{B}%
			_{i}^{\left( \ell \right) }}{B_{i}^{\left( \ell \right) }}}\right] ,  \notag
	\\
	&\sim &\sum_{\ell =1}^{L}\frac{\psi _{\ell }}{\phi _{\ell }}\sum_{i=1}^{m}%
	\mathcal{F}_{\ell ,i}\left( \frac{\phi _{\ell }\tau _{\varkappa }\gamma _{%
			\text{th}}^{\left( \varkappa \right) }}{\omega _{\varkappa }-\varrho
		_{\varkappa }\gamma _{\text{th}}^{\left( \varkappa \right) }}\right) ^{\frac{%
			\mathcal{B}_{i}^{\left( \ell \right) }}{B_{i}^{\left( \ell \right) }}}.
	\label{asymptopoutnonideal}
	\end{eqnarray}
	
	\begin{remark}
		For ideal RF front e2e ( i.e., $\omega _{\varkappa }=\tau
		_{\varkappa }=1$ and $\varrho _{\varkappa }=0)$, the asymptotic
		expression of OP, for any arbitrary values of $\gamma _{\text{th}}$, can be
		reexpressed as
		\begin{equation}
		P_{\text{out}}^{\left( id\right) }\sim \sum_{\ell =1}^{L}\frac{\psi _{\ell }%
		}{\phi _{\ell }}\sum_{i=1}^{m}\mathcal{F}_{\ell ,i}\left( \phi _{\ell
		}\gamma _{\text{th}}\right) ^{\frac{\mathcal{B}_{i}^{\left( \ell \right) }}{%
				B_{i}^{\left( \ell \right) }}}.  \label{asymptopoutideal}
		\end{equation}
	\end{remark}
	\subsection{Channel Capacity}
	
	\subsubsection{Exact analysis}
	
	The ergodic channel capacity, for a given impairment $\varkappa $ in bits/s,
	under constant power with ORA policy over any composite fading or shadowing
	channel is expressed as
	
	\begin{equation}
	\mathcal{C}^{\left( \varkappa \right) }=\frac{1}{\log (2)}\int_{0}^{\frac{%
			\omega _{\varkappa }}{\varrho _{\varkappa }}}\log (1+\gamma )\frac{\partial
		F_{\gamma _{\varkappa }}\left( \gamma \right) }{\partial \gamma }d\gamma .
	\label{ORAexpression}
	\end{equation}
	
	\begin{proposition}
		\label{ORAexact} The
		closed-form expression of the CC under ORA policy can be
		written as
		\begin{IEEEeqnarray}{rCl}
			\mathcal{C}^{\left( \varkappa \right) }=\frac{1}{\log \left( 2\right) }%
			\sum_{\ell =1}^{L}\frac{\psi _{\ell }}{\phi _{\ell }}{\small H}_{1,0:p%
				{\small ,q+1:1,2}}^{0,1:m+1{\small ,n:1,1}}\left( \frac{\phi _{\ell }\tau
				_{\varkappa }}{\varrho _{\varkappa }},\frac{\omega _{\varkappa }}{\varrho
				_{\varkappa }}\left\vert
			\begin{array}{c}
				\left( {\small 1};{\small 1},{\small 1}\right) \\
				-%
			\end{array}%
			\right\vert
			\begin{array}{c}
				\Psi ^{\left( \ell \right) } \\
				\left( 0,1\right) ,\Upsilon ^{\left( \ell \right) }%
			\end{array}%
			\left\vert
			\begin{array}{c}
				\left( 1,1\right) ;- \\
				\left( 1,1\right) ;\left( 0,1\right)%
			\end{array}%
			\right. \right).  \label{ORAcapa}
		\end{IEEEeqnarray}
		where ${\small H}_{p_{1},q_{1}:p_{2},q_{2}:p_{3},q_{3}}^{m_{1}{\small ,n}_{1}%
			{\small :m}_{2}{\small ,n}_{2}{\small :m}_{3},n_{3}}\left( .\right) $
		denotes the Bivariate FHF \cite[Eq. (2.57)]{mathai}.
	\end{proposition}
	
	\begin{proof}
		The proof is provided in Appendix A.
	\end{proof}
	
	\begin{corollary}
		\textcolor{black}{For ideal RF e2e, the expression of the CC under
			ORA policy can be reduced to}
		\begin{equation}
		\mathcal{C}^{\left( id\right) }=\frac{1}{\log \left( 2\right) }\sum_{\ell
			=1}^{L}\frac{\psi _{\ell }}{\phi _{\ell }}{\small H}_{p+2{\small ,q+2}}^{m+2%
			{\small ,n+1}}\left( \phi _{\ell }\left\vert
		\begin{array}{c}
		\left( 0,1\right) ,\Psi ^{\left( \ell \right) },\left( 1,1\right) \\
		\left( 0,1\right) ,\left( 0,1\right) ,\Upsilon ^{\left( \ell \right) }%
		\end{array}%
		\right. \right).  \label{idealcapa}
		\end{equation}%
		\label{crolaryidealcapa}
	\end{corollary}
	
	\begin{proof}
		The proof is provided in Appendix B.
	\end{proof}
	
	\subsubsection{Asymptotic analysis}
	
	By plugging $\left( \ref{ApproFHF}\right) $ in $\left( \ref{EQORAappedix}%
	\right) $, the asymptotic expression of the \textcolor{black}{CC} in high SNR
	regime can be written as
	\begin{eqnarray}
	\mathcal{C}^{\left( \varkappa \right) } &\sim &\log _{2}\left( 1+\frac{%
		\omega _{\varkappa }}{\varrho _{\varkappa }}\right) \sum_{\ell =1}^{L}\frac{%
		\psi _{\ell }}{\phi _{\ell }}E_{\ell }-\sum_{\ell =1}^{L}\frac{\psi _{\ell }%
	}{\log \left( 2\right) \phi _{\ell }}\sum_{i=1}^{m}\mathcal{F}_{\ell ,i}\int_{0}^{\frac{\omega _{\varkappa }}{\varrho _{\varkappa }}}\frac{1%
	}{1+\gamma }\left( \frac{\phi _{\ell }\tau _{\varkappa }\gamma }{\omega
		_{\varkappa }-\varrho _{\varkappa }\gamma }\right) ^{\frac{\mathcal{B}%
			_{i}^{\left( \ell \right) }}{B_{i}^{\left( \ell \right) }}}d\gamma ,  \notag
	\\
	&\overset{(a)}{\sim }&\log _{2}\left( 1+\frac{\omega _{\varkappa }}{\varrho
		_{\varkappa }}\right) -\sum_{\ell =1}^{L}\frac{\psi _{\ell }}{\log \left(
		2\right) \phi _{\ell }}\sum_{i=1}^{m}\mathcal{F}_{\ell ,i}\int_{0}^{\frac{\omega _{\varkappa }}{\varrho _{\varkappa }}}\frac{1%
	}{1+\gamma }\left( \frac{\phi _{\ell }\tau _{\varkappa }\gamma }{\omega
		_{\varkappa }-\varrho _{\varkappa }\gamma }\right) ^{\frac{\mathcal{B}%
			_{i}^{\left( \ell \right) }}{B_{i}^{\left( \ell \right) }}}d\gamma .
	\end{eqnarray}%
	where step $\left( a\right) $ holds using $\left( \ref{Condition}\right) $.
	By using \cite[Eqs. (2.9.5),(2.9.15)]{kilbas} and \cite[3.194 Eq. (1)]%
	{integraltable}, one obtains
	\begin{eqnarray}
	\mathcal{C}^{\left( \varkappa \right) } &\sim &\log _{2}\left( 1+\frac{%
		\omega _{\varkappa }}{\varrho _{\varkappa }}\right) -\frac{\omega
		_{\varkappa }}{\log (2)\varrho _{\varkappa }}\sum_{\ell =1}^{L}\frac{\psi
		_{\ell }}{\phi _{\ell }}\sum_{i=1}^{m}\frac{\mathcal{F}_{\ell ,i}}{\Gamma
		\left( \frac{\mathcal{B}_{i}}{Bi}\right) }\left( \frac{\phi _{\ell }\tau
		_{\varkappa }}{\varrho _{\varkappa }}\right) ^{\frac{\mathcal{B}_{i}^{\left(
				\ell \right) }}{B_{i}^{\left( \ell \right) }}}\left( \frac{1}{2\pi j}\right)
	^{2}\int_{\mathcal{L}_{s}^{\left( \ell \right) }}\int_{\mathcal{L}%
		_{v}^{\left( \ell \right) }}\Gamma \left( s\right)  \notag \\
	&&\times \frac{\Gamma \left( 1-s\right) \Gamma \left( v\right) \Gamma \left(
		\frac{\mathcal{B}_{i}^{\left( \ell \right) }}{B_{i}^{\left( \ell \right) }}%
		-v\right) \Gamma \left( 1+\frac{\mathcal{B}_{i}^{\left( \ell \right) }}{%
			B_{i}^{\left( \ell \right) }}-s-v\right) }{\Gamma \left( 2+\frac{\mathcal{B}%
			_{i}^{\left( \ell \right) }}{B_{i}^{\left( \ell \right) }}-s-v\right) }%
	\left( \frac{\omega _{\varkappa }}{\varrho _{\varkappa }}\right) ^{-s}\left(
	-1\right) ^{-v}dsdv.  \label{ZZ}
	\end{eqnarray}%
	which can be equivalently expressed using the bivariate FHF as
	\begin{IEEEeqnarray}{rCl}
		\mathcal{C}^{\left( \varkappa \right) } &\sim &\log _{2}\left( 1+\frac{%
			\omega _{\varkappa }}{\varrho _{\varkappa }}\right) -\frac{\omega
			_{\varkappa }}{\log (2)\varrho _{\varkappa }}\sum_{\ell =1}^{L}\frac{\psi
			_{\ell }}{\phi _{\ell }}\sum_{i=1}^{m}\frac{\mathcal{F}_{\ell ,i}}{\Gamma
			\left( \frac{\mathcal{B}_{i}}{Bi}\right) }\left( \frac{\phi _{\ell }\tau
			_{\varkappa }}{\varrho _{\varkappa }}\right) ^{\substack{ \frac{\mathcal{B}%
					_{i}^{\left( \ell \right) }}{B_{i}^{\left( \ell \right) }} \\ }}  \notag \\
		&&\times {\small H}_{1,1:1{\small ,1:1,1}}^{0,1:1{\small ,1:1,1}}\left(
		\frac{\omega _{\varkappa }}{\varrho _{\varkappa }},-1\left\vert
		\begin{array}{c}
			\left( {\small -}\frac{\mathcal{B}_{i}^{\left( \ell \right) }}{B_{i}^{\left(
					\ell \right) }};{\small 1},{\small 1}\right) ;- \\
			-;\left( -1{\small -}\frac{\mathcal{B}_{i}^{\left( \ell \right) }}{%
				B_{i}^{\left( \ell \right) }};{\small 1},{\small 1}\right)
		\end{array}%
		\right\vert
		\begin{array}{c}
			\left( 0,1\right) ;- \\
			\left( 0,1\right) ;-%
		\end{array}%
		\left\vert
		\begin{array}{c}
			\left( 1-\frac{\mathcal{B}_{i}^{\left( \ell \right) }}{B_{i}^{\left( \ell
					\right) }},1\right) ;- \\
			\left( 0,1\right) ;-%
		\end{array}%
		\right. \right) .  \label{asymptoORA}
	\end{IEEEeqnarray}
	
	\begin{remark}\label{remarkCapaceiling}
		\begin{itemize}\item[]
			\item As $\phi _{\ell }$ is inversely proportional to $\overline{\gamma }%
			_{id}$ and based on $\left( \ref{Condition}\right) $,
			\textcolor{black}{it
				follows} that $\frac{\psi _{\ell }}{\phi _{\ell }}\neq \infty .$ Thus, it
			can be clearly seen from $\left( \ref{asymptoORA}\right)$ that
			\begin{equation}
			\lim_{\overline{\gamma }_{id}\rightarrow \infty }\mathcal{C}^{\left(
				\varkappa \right) }\mathcal{=C}^{\left( \varkappa ,\infty \right) }=\log
			_{2}\left( 1+\frac{\omega _{\varkappa }}{\varrho _{\varkappa }}\right) ,
			\label{capaceiling}
			\end{equation}%
			which demonstrates that the capacity channel for non-ideal RF front end
			(i.e., $\varrho _{\varkappa }\neq 0$) has a ceiling, depending exclusively
			on the impairments parameters and is irrespective of the fading severity
			parameters, that can't be crossed by increasing the SNR.
			
			\item The asymptotic expression of the CC under ideal case can be obtained
			straightforwardly using \cite[Eq. (1.8.2)]{kilbas} \textcolor{black}{along}
			with (\ref{idealcapa}).
		\end{itemize}
	\end{remark}
	\subsection{Average Symbol Error Probability}
	
	In this section, a tight approximate expression for the ASEP of various
	modulation schemes is derived. We start first by obtaining a simple
	exponential-based approximate expression for the SEP using the Trapezoidal
	integration rule \cite{yassineaccess}.
	
	Let $\mathcal{H}\left( \gamma \right)$ denote the SEP for a given modulation
	scheme. Table \ref{SEPmodulations} outlines the SEP's integral-based
	expression for different modulation techniques.
	\begin{table}[tbh]
		\captionsetup{font=small}
		\caption{Integral form of the SEP for numerous modulation schemes.}
		\label{SEPmodulations}\centering%
		\begin{tabular}{c|c}
			\hline
			\textit{Modulation scheme} & $\mathcal{H}\left( \gamma \right) $ \\
			\hline\hline
			$M$-PSK & $\mathcal{D}\left( 1,\frac{M-1}{M},\sin ^{2}\left( \frac{\pi }{M}%
			\right) ,\gamma \right) $ \\ \hline
			$M$-QAM & $\mathcal{D}\left( \frac{4\left( \sqrt{M}-1\right) }{\sqrt{M}},%
			\frac{1}{2},\frac{3}{2\left( M-1\right) },\gamma \right) -\mathcal{D}\left(
			4\left( \frac{\sqrt{M}-1}{\sqrt{M}}\right) ^{2},\frac{1}{4},\frac{3}{2\left(
				M-1\right) },\gamma \right) $ \\ \hline
			$M$-DPSK & $\mathcal{E}\left( 1,\frac{M-1}{M},\sin ^{2}\left( \frac{\pi }{M}%
			\right) ,\cos \left( \frac{\pi }{M}\right) ,\gamma \right) $ \\ \hline
			GC-DQPSK & $\mathcal{E}\left( \frac{1}{2},1,1,-\frac{1}{\sqrt{2}},\gamma
			\right) $ \\ \hline\hline
		\end{tabular}%
	\end{table}
	\par
	with
	\begin{equation}
	\mathcal{D}\left( a,b,c,\gamma \right) =\frac{a}{\pi }\int_{0}^{b\pi }\exp
	\left( -\frac{c\gamma }{\sin ^{2}\left( t\right) }\right) dt,  \label{Dfunc}
	\end{equation}%
	and
	\begin{equation}
	\mathcal{E}\left( a,b,c,d,\gamma \right) =\frac{a}{\pi }\int_{0}^{b\pi }\exp
	\left( -\frac{c\gamma }{1+d\cos \left( t\right) }\right) dt.  \label{Efunc}
	\end{equation}%
	\newline
	It follows that the ASEP can be calculated by averaging $\mathcal{H}\left(
	\gamma \right) $ over the statistics of the involved fading channel as
	\begin{equation}
	P_{s}^{\left( \varkappa \right) }=\int_{0}^{\infty }f_{\gamma _{\varkappa
	}}\left( \gamma \right) \mathcal{H}\left( \gamma \right) d\gamma .
	\label{generalASEP}
	\end{equation}%
	In the sequel, \textcolor{black}{a} tight approximate expression for the SEP
	per each considered modulation scheme is obtained, relying on Table \ref%
	{SEPmodulations} alongside the Trapezoidal rule. Leveraging these results,
	accurate approximate and asymptotic expressions for ASEP are provided using $%
	\left( \ref{pdf}\right) $ jointly with $\left( \ref{generalASEP}\right) $.
	
	\begin{proposition}
		The SEP for the considered modulation schemes can be tightly approximated by
		\begin{equation}
		\mathcal{H}\left( \gamma \right) \simeq \sum_{n=0}^{N}\theta _{n}\exp \left(
		-\delta _{n}\gamma \right) ,  \label{SEPbin}
		\end{equation}%
		where $\theta _{n}$ and $\delta _{n}$ are summarized in Tables \ref%
		{diracandteta1} and \ref{diracandteta2} depending on the %
		\textcolor{black}{employed} modulation scheme. Note that an even positive
		number $N$ is required to evaluate such coefficients for $M$-QAM modulation
		technique.
		\begin{table}[tbh]
			\captionsetup{font=small,width=0.8\columnwidth}
			\caption{The coefficients $\protect\theta _{n}$ and $\protect%
				\delta _{n}$ for $M$-PSK, $M$-DPSK and GC-DQPSK modulations.}
			\label{diracandteta1}\renewcommand{\arraystretch}{1} \centering%
			\begin{tabular}[b]{c|c|c|c}
				\hline
				\textit{Modulation} & \multicolumn{2}{c|}{$\theta _{n}$} & %
				\multirow{2}{*}{\textcolor{black}{$\delta _{n}\ (0\leq n\leq N)$}} \\ \cline{2-3}
				\textit{scheme} & $n=0$, $N$ & $1\leq n\leq N-1$ &  \\ \hline\hline
				$M$-PSK & $\frac{M-1}{2NM}$ & $\frac{M-1}{NM}$ & $\frac{\sin ^{2}\left(
					\frac{\pi }{M}\right) }{\sin ^{2}\left( \frac{n\left( M-1\right) \pi }{NM}%
					\right) }$ \\ \hline
				$M$-DPSK & $\frac{M-1}{2NM}$ & $\frac{M-1}{NM}$ & $\frac{\sin ^{2}\left(
					\frac{\pi }{M}\right) }{1+\cos \left( \frac{\pi }{M}\right) \cos \left(
					\frac{n\left( M-1\right) \pi }{NM}\right) }$ \\ \hline
				GC-DQPSK & $\frac{1}{4N}$ & $\frac{1}{2N}$ & $\frac{1}{1-\frac{1}{\sqrt{2}}%
					\cos \left( \frac{n\pi }{N}\right) }$ \\ \hline\hline
			\end{tabular}%
		\end{table}
		\begin{table}[tbh]
			\captionsetup{font=small}
			\caption{The coefficients $\protect\theta _{n}$ and $\protect%
				\delta _{n}$ for $M$-QAM modulation.}
			\label{diracandteta2}\renewcommand{\arraystretch}{1} \centering%
			\begin{tabular}[b]{c|c|c|c|c|c|c}
				\hline
				\textit{Modulation} & \multicolumn{5}{c|}{$\theta _{n}$} & %
				\multirow{2}{*}{\textcolor{black}{$\delta _{n}\ (0\leq n\leq N)$}} \\ \cline{2-6}
				\textit{scheme} & $n=0$ & $1\leq n\leq \frac{N}{2}-1$ & $n=\frac{N}{2}$ & $%
				\frac{N}{2}+1\leq n\leq N-1$ & $n=N$ &  \\ \hline\hline
				$M$-QAM & $\frac{\sqrt{M}-1}{NM}$ & $2\frac{\sqrt{M}-1}{NM}$ & $\frac{M-1}{NM%
				}$ & $2\frac{\sqrt{M}-1}{N\sqrt{M}}$ & $2\frac{\sqrt{M}-1}{N\sqrt{M}}$ & $%
				\frac{3}{2\left( M-1\right) \sin ^{2}\left( \frac{n\pi }{2N}\right) }$ \\
				\hline\hline
			\end{tabular}%
		\end{table}
	\end{proposition}
	
	\begin{proof}
		Importantly, for a given positive number $N$, a real-valued $\varphi $, and
		an arbitrary function $f$, the following definite integral is known to be
		accurately approximated using the Trapezoidal rule as
		\begin{equation}
		\int_{0}^{N\varphi }f\left( t\right) dt=\frac{\varphi }{2}\left[
		g_{0}+g_{N}+2\sum\limits_{n=1}^{N-1}g_{n}\right] .  \label{TrapzEq}
		\end{equation}%
		where $g_{0}=f\left( 0\right) $ and $g_{n}=f\left( n\varphi \right) .$ By
		setting $\varphi =\frac{b\pi }{N}$ and \textcolor{black}{the function $f$ in} $%
		\left( \ref{TrapzEq}\right) $ to either $\mathcal{D}$ and $\mathcal{E}$
		defined in $\left( \ref{Dfunc}\right) $ or $\left( \ref{Efunc}\right) $,
		respectively, along with Table \ref{SEPmodulations}'s coefficients, and
		performing some algebraic operations, both Tables $\ref{diracandteta1},$ and
		$\ref{diracandteta2},$ can be easily obtained, which concludes the proof.
	\end{proof}
	
	\begin{remark}
		Of note, the greater $N$ is, the accurate the approximation is. However, the
		computational complexity becomes higher with the increase of such a number.
		To this end, it is necessary to look for an accuracy-complexity tradeoff.
		Numerically, we checked that above $N\geq 5$, the relative error becomes
		negligible \cite{yassineaccess}. Owing to this fact, $N$ is set to $5$ for
		the modulation techniques summarized in Table $\ref{diracandteta1}$ and to $%
		6 $ (i.e. an even number above $5$) for $M$-QAM.
	\end{remark}
	
	\subsubsection{Exact Analysis}
	
	\begin{proposition}
		\label{propASEPexact}
		The ASEP for the four considered modulation techniques can be tightly
		approximated as
		\begin{IEEEeqnarray}{rCl}
			P_{s}^{\left( \varkappa \right) } &\simeq &\sum_{n=0}^{N}\theta _{n}\left[
			1-\sum_{\ell =1}^{L}\frac{\psi _{\ell }}{\phi _{\ell }}{\small H}_{1,0:p%
				{\small ,q+1:0,2}}^{0,1:m+1{\small ,n:1,0}}\left( \frac{\phi _{\ell }\tau
				_{\varkappa }}{\varrho _{\varkappa }},\frac{\delta _{n}\omega _{\varkappa }}{%
				\varrho _{\varkappa }}\left\vert
			\begin{array}{c}
				\left( {\small 1};{\small 1},{\small 1}\right)  \\
				-%
			\end{array}%
			\right\vert
			\begin{array}{c}
				\Psi ^{\left( \ell \right) } \\
				\left( 0,1\right) ,\Upsilon ^{\left( \ell \right) }%
			\end{array}%
			\left\vert
			\begin{array}{c}
				- \\
				\left( 1,1\right) ;\left( 0,1\right)
			\end{array}%
			\right. \right) \right] .  \notag \\
			&&  \label{ASEPMARY}
		\end{IEEEeqnarray}
	\end{proposition}
	
	\begin{proof}
		The proof is provided in Appendix C.
	\end{proof}
	
	\begin{corollary}
		\label{corolidealasep}
		For \textcolor{black}{ideal RF e2e}, the above approximate ASEP's expression
		can be simplified in terms of a univariate FHF instead as
		\begin{equation}
		P_{s}^{\left( \varkappa \right) }\simeq \sum_{n=1}^{N}\theta _{n}-\sum_{\ell
			=1}^{L}\frac{\psi _{\ell }}{\phi _{\ell }}\sum_{n=1}^{N}\theta _{n}{\small H}%
		_{p+2{\small ,q+1}}^{m+1{\small ,n+1}}\left( \frac{\phi _{\ell }}{\delta _{n}%
		}\left\vert
		\begin{array}{c}
		\left( 0,1\right) ,\Psi ^{\left( \ell \right) },\left( 1,1\right)  \\
		\left( 0,1\right) ,\Upsilon ^{\left( \ell \right) }%
		\end{array}%
		\right. \right) .  \label{idealASEP}
		\end{equation}
	\end{corollary}
	
	\begin{proof}
		The proof is provided in Appendix D.
	\end{proof}
	
	\subsubsection{Asymptotic analysis}
	
	\begin{proposition}
		\label{propasympAsep}
		the asymptotic expression of the ASEP can be expressed as
		\begin{IEEEeqnarray}{rCl}
			P_{s}^{\left( \varkappa \right) } &\sim &\sum_{n=1}^{N}\theta _{n}\exp
			\left( -\frac{\delta _{n}\omega _{\varkappa }}{\varrho _{\varkappa }}\right)
			+\frac{\omega _{\varkappa }}{\varrho _{\varkappa }}\sum_{\ell =1}^{L}\frac{%
				\psi _{\ell }}{\phi _{\ell }}\sum_{n=1}^{N}\theta _{n}\delta
			_{n}\sum_{i=1}^{m}\frac{\mathcal{F}_{\ell ,i}\left( \frac{\phi _{\ell }\tau
					_{\varkappa }}{\omega _{\varkappa }}\right) }{\Gamma \left( \frac{\mathcal{B}%
					_{i}}{B_{i}}\right) }^{\frac{\mathcal{B}_{i}^{\left( \ell \right) }}{%
					B_{i}^{\left( \ell \right) }}}  \notag \\
			&&\times {\small H}_{1,1:0{\small ,1:1,1}}^{0,1:1{\small ,0:1,1}}\left(
			\frac{\delta _{n}\omega _{\varkappa }}{\varrho _{\varkappa }},-1\left\vert
			\begin{array}{c}
				\left( {\small -}\frac{\mathcal{B}_{i}^{\left( \ell \right) }}{B_{i}^{\left(
						\ell \right) }};{\small 1},{\small 1}\right)  \\
				-;\left( -1{\small -}\frac{\mathcal{B}_{i}^{\left( \ell \right) }}{%
					B_{i}^{\left( \ell \right) }};{\small 1},{\small 1}\right)
			\end{array}%
			\right\vert
			\begin{array}{c}
				- \\
				\left( 0,1\right) ;-%
			\end{array}%
			\left\vert
			\begin{array}{c}
				\left( 1-\frac{\mathcal{B}_{i}^{\left( \ell \right) }}{B_{i}^{\left( \ell
						\right) }},1\right) ;- \\
				\left( 0,1\right) ;-%
			\end{array}%
			\right. \right) .  \label{AsymptoASEP}
		\end{IEEEeqnarray}
	\end{proposition}
	
	\begin{proof}
		the proof is provided in appendix E.
	\end{proof}
	
	\begin{remark}
		\begin{itemize}
			\item[]
			\item Similarly to (\ref{capaceiling}), for \textcolor{black}{non-ideal RF e2e}%
			, the ASEP admits the following ceiling that can't be crossed by the
			increase of SNR as%
			\begin{equation}
			\lim_{_{\overline{\gamma }_{id}\rightarrow \infty }}P_{s}^{\left( \varkappa
				\right) }=P_{s}^{\left( \varkappa ,\infty \right) }=\sum_{n=0}^{N}\theta
			_{n}\exp \left( -\frac{\delta _{n}\omega _{\varkappa }}{\varrho _{\varkappa }%
			}\right) .  \label{valofASEPatinfinity}
			\end{equation}%
			Moreover, it can be seen that this ceiling is irrespective of the fading
			severity parameters, whereas it depends on both $M$ (as $\delta _{n}$ and $%
			\theta _{n}$ are expressed in terms of the modulation %
			\textcolor{black}{parameter} $M$) and the impairments parameters.
			Particularly, when $\frac{\omega _{\varkappa }}{\varrho _{\varkappa }}$
			approaches infinity, $P_{s}^{\left( \varkappa \right) }$ tends to $0$. On
			the other hand, it can be easily shown that $\theta _{n}$ and $\delta _{n}$
			are increasing and decreasing functions with $M$, respectively. As a result,
			$\theta _{n}\exp \left( -\frac{\delta _{n}\omega _{\varkappa }}{\varrho
				_{\varkappa }}\right) $ is an increasing function \textcolor{black}{with $M$}.
			Thus the smaller the value of $M$ is, the smaller the ASEP is, leading to
			the system performance enhancement.
			\item For ideal case, the ASEP's asymptotic expression can be easily
			obtained directly from (\ref{idealASEP}) with the help of \cite[Eq. (1.8.2)]%
			{kilbas}.
		\end{itemize}
	\end{remark}
	\section{Applications}
	
	It is worthwhile to note that the aforementioned derived expressions for
	various metrics can be simplified significantly depending on the fading amplitude distribution simplicity. As mentioned before, \textcolor{black}{numerous} well-known PDF of recent fading distributions can be expressed in terms of FHF. In this
	section, two generalized fading models are considered:
	
	\subsection{$\protect\alpha-\protect\mu$ fading}
	
	Let's $\alpha $ and $\mu $ denote two real numbers that reflect the
	non-linearity of the propagation medium and the clustering of the multipath
	waves, \textcolor{black}{respectively}. The SNR's PDF of $\alpha -\mu $ fading
	model can be obtained by setting $L=1$, $\psi _{1}=$ $\frac{\mu _{{}}^{\frac{%
				2}{\alpha }}}{\overline{\gamma }\Gamma \left( \mu \right) }$, $\phi _{1}=%
	\frac{\mu ^{\frac{2}{\alpha }}}{\overline{\gamma }}$, $m=q=1,n=p=0$, $%
	b_{1}^{\left( 1\right) }=$ $\mu -\frac{2}{\alpha }$, and $B_{1}^{\left(
		1\right) }=\frac{2}{\alpha }$ in (\ref{pdf}). Therefore, the closed-form and
	asymptotic expressions for the three considered metrics are summarized in
	Tables \ref{closedalpha} and \ref{asympalpha}, respectively, where $\Psi
	\left( \cdot \right) $ denotes the Digamma function which is defined as the
	logarithmic derivative of the Gamma function \cite{digam} and
	\begin{equation}
	\Theta _{\varkappa }=\frac{\omega _{\varkappa }\mu ^{\mu }}{\varrho
		_{\varkappa }\Gamma \left( \frac{\alpha \mu }{2}\right) \Gamma \left( 1+\mu
		\right) }\left( \frac{\tau _{\varkappa }}{\varrho _{\varkappa }\overline{%
			\gamma }_{id}}\right) ^{\frac{\alpha \mu }{2}}.
	\end{equation}
	\begin{table}[tbh]
		\captionsetup{font=small}
		\caption{Closed-form expressions for $\protect\alpha -\protect\mu $ fading
			model.}\renewcommand{\arraystretch}{1.2} \centering%
		\begin{tabular}{c|c|c}
			\hline
			\parbox[t]{2mm}{\multirow{5}{*}{\rotatebox[origin=c]{90}{\small{%
							\textit{Ideal RF e2e}}}}} & OP & $1-\frac{1}{\Gamma \left( \mu \right) }%
			H_{1,2}^{2,0}\left( \frac{\mu ^{\frac{2}{\alpha }}}{\overline{\gamma }_{id}}%
			\gamma _{\text{th}}\left\vert
			\begin{array}{c}
			-;\left( 1,1\right)  \\
			\left( 0,1\right) ,\left( \mu ,\frac{2}{\alpha }\right) ;-%
			\end{array}%
			\right. \right) $ \\ \cline{2-3}
			& CC & $\frac{1}{\log (2)\Gamma \left( \mu \right) }H_{2,3}^{3,1}\left(
			\frac{\mu ^{\frac{2}{\alpha }}}{\overline{\gamma }_{id}}\left\vert
			\begin{array}{c}
			\left( 0,1\right) ;\left( 1,1\right)  \\
			\left( 0,1\right) ,\left( 0,1\right) ,\left( \mu ,\frac{2}{\alpha }\right) ;-%
			\end{array}%
			\right. \right) $ \\ \cline{2-3}
			& ASEP & $\sum_{n=0}^{N}\theta _{n}\left[ 1-\frac{1}{\Gamma \left( \mu
				\right) }H_{2,2}^{2,1}\left( \frac{\mu ^{\frac{2}{\alpha }}}{\delta _{n}%
				\overline{\gamma }_{id}}\left\vert
			\begin{array}{c}
			\left( 0,1\right) ;\left( 1,1\right)  \\
			\left( 0,1\right) ,\left( \mu ,\frac{2}{\alpha }\right) ;-%
			\end{array}%
			\right. \right) \right] $ \\ \hline\hline
			\parbox[t]{2mm}{\centering
				\multirow{5}{*}{\rotatebox[origin=c]{90}{\small{\textit{Non-ideal RF e2e}}}}}
			& OP & $1-\frac{1}{\Gamma \left( \mu \right) }H_{1,2}^{2,0}\left( \frac{\mu
				^{\frac{2}{\alpha }}}{\overline{\gamma }_{id}}\frac{\tau _{\varkappa }\gamma
				_{\text{th}}}{\omega _{\varkappa }-\varrho _{\varkappa }\gamma _{\text{th}}}%
			\left\vert
			\begin{array}{c}
			-;\left( 1,1\right)  \\
			\left( 0,1\right) ,\left( \mu ,\frac{2}{\alpha }\right) ;-%
			\end{array}%
			\right. \right) $ \\ \cline{2-3}
			& CC & $\frac{1}{\log (2)\Gamma \left( \mu \right) }%
			H_{1,0:0,2:1,2}^{0,1:2,0:1,1}\left( \frac{\mu ^{\frac{2}{\alpha }}\tau
				_{\varkappa }}{\overline{\gamma }_{id}\varrho _{\varkappa }},\frac{\omega
				_{\varkappa }}{\varrho _{\varkappa }}\left\vert
			\begin{array}{c}
			\left( 1;1,1\right)  \\
			-%
			\end{array}%
			\right\vert
			\begin{array}{c}
			- \\
			\left( 0,1\right) ,\left( \mu ,\frac{2}{\alpha }\right) ;-%
			\end{array}%
			\left\vert
			\begin{array}{c}
			\left( 1,1\right) ;- \\
			\left( 1,1\right) ;\left( 0,1\right)
			\end{array}%
			\right. \right) $ \\ \cline{2-3}
			& ASEP & $\sum_{n=0}^{N}\theta _{n}\left[ 1-\frac{1}{\Gamma \left( \mu
				\right) }H_{1,0:0,2:0,2}^{0,1:2,0:1,0}\left( \frac{\mu ^{\frac{2}{\alpha }%
				}\tau _{\varkappa }}{\overline{\gamma }_{id}\varrho _{\varkappa }},\frac{%
				\delta _{n}\omega _{\varkappa }}{\varrho _{\varkappa }}\left\vert
			\begin{array}{c}
			\left( 1;1,1\right)  \\
			-%
			\end{array}%
			\right\vert
			\begin{array}{c}
			- \\
			\left( 0,1\right) ,\left( \mu ,\frac{2}{\alpha }\right) ;-%
			\end{array}%
			\left\vert
			\begin{array}{c}
			- \\
			\left( 1,1\right) ;\left( 0,1\right)
			\end{array}%
			\right. \right) \right] $ \\ \hline
		\end{tabular}%
		\label{closedalpha}
	\end{table}
	\begin{table}[tbh]
		\captionsetup{font=small}
		\caption{Asymptotic expressions for $\protect\alpha -\protect\mu $ fading
			model.}
		\label{asympalpha}\centering\renewcommand{\arraystretch}{2}%
		\begin{tabular}{c|c|c}
			\hline
			\parbox[t]{2mm}{\multirow{3}{*}{\rotatebox[origin=c]{90}{\small{%
							\textit{Ideal RF e2e}}}}} & OP & $\frac{\mu ^{\mu -1}}{\Gamma \left( \mu
				\right) }\left( \frac{\gamma _{\text{th}}}{\overline{\gamma }_{id}}\right) ^{%
				\frac{\alpha \mu }{2}} $ \\ \cline{2-3}
			& CC & $\frac{1}{\log (2)}\left[ \frac{2}{\alpha }\Psi \left( \mu \right)
			-\log \left( \frac{\mu ^{\frac{2}{\alpha }}}{\overline{\gamma }_{id}}\right) %
			\right] $ \\ \cline{2-3}
			& ASEP & $\frac{\alpha \Gamma \left( \frac{\alpha \mu }{2}\right) \mu ^{\mu }%
			}{2\Gamma \left( \mu \right) }\sum_{n=0}^{N}\theta _{n}\left( \frac{1}{%
				\delta _{n}\overline{\gamma }_{id}}\right) ^{\frac{\alpha \mu }{2}}$ \\
			\hline\hline
			\parbox[t]{2mm}{\multirow{5}{*}{\rotatebox[origin=c]{90}{\small{%
							\textit{Non-ideal RF e2e}}}}} & OP & $\frac{\mu ^{\mu -1}}{\Gamma \left( \mu
				\right) }\left( \frac{\tau _{\varkappa }\gamma _{\text{th}}}{\left( \omega
				_{\varkappa }-\varrho _{\varkappa }\gamma _{\text{th}}\right) \overline{%
					\gamma }_{id}}\right) ^{\frac{\alpha \mu }{2}}$ \\ \cline{2-3}
			& CC & $C^{\left( \varkappa ,\infty \right) }-\frac{\Theta _{\varkappa }}{%
				\log (2)}H_{1,1:1,1:1,1}^{0,1:1,1:1,1}\left( \frac{\omega _{\varkappa }}{%
				\varrho _{\varkappa }},-1\left\vert
			\begin{array}{c}
			\left( -\frac{\alpha \mu }{2};1,1\right) ;- \\
			-;\left( -1-\frac{\alpha \mu }{2};1,1\right)%
			\end{array}%
			\right\vert
			\begin{array}{c}
			\left( 0,1\right) ;- \\
			\left( 0,1\right) ;-%
			\end{array}%
			\left\vert
			\begin{array}{c}
			\left( 1-\frac{\alpha \mu }{2},1\right) ;- \\
			\left( 0,1\right) ;-%
			\end{array}%
			\right. \right) $ \\ \cline{2-3}
			& ASEP & $P_{s}^{\left( \varkappa ,\infty \right) }+\Theta _{\varkappa
			}\sum_{n=1}^{N}\theta _{n}\delta _{n}H_{1,1:0,1:1,1}^{0,1:1,0:1,1}\left(
			\frac{\delta _{n}\omega _{\varkappa }}{\varrho _{\varkappa }},-1\left\vert
			\begin{array}{c}
			\left( -\frac{\alpha \mu }{2};1,1\right) \\
			-;\left( -1-\frac{\alpha \mu }{2};1,1\right)%
			\end{array}%
			\right\vert
			\begin{array}{c}
			- \\
			\left( 0,1\right) ;-%
			\end{array}%
			\left\vert
			\begin{array}{c}
			\left( 1-\frac{\alpha \mu }{2},1\right) ;- \\
			\left( 0,1\right) ;-%
			\end{array}%
			\right. \right) $ \\ \hline
		\end{tabular}%
	\end{table}
	\subsection{M\'alaga $\mathcal{M}$ \textcolor{black}{TCPE}}
	According to \cite{faissal}, for M\'alaga $\mathcal{M}$ TCPE, the SNR's \ PDF
	can be expressed in terms of FHF by setting $L=\beta ,$ $\phi _{\ell }=\frac{%
		B^{\prime r}}{\mu _{r}}$, $\psi _{\ell }=\xi ^{2}A^{\prime }\lambda _{\ell
	}\phi _{\ell },$ $m=q=3$, $n=0$, $p=1,$ $b_{1}^{\left( \ell \right) }=\xi
	^{2}-r,$ $b_{2}^{\left( \ell \right) }=\alpha -r,$ $b_{3}^{\left( \ell
		\right) }=\ell -r$ and $B_{1}^{\left( \ell \right) }=B_{2}^{\left( \ell
		\right) }=B_{3}^{\left( \ell \right) }=r$ while $a_{1}^{\left( \ell \right)
	}=\xi ^{2}-r+1$ and $\ A_{1}^{\left( \ell \right) }=r$ in $\left( \ref{pdf}%
	\right) $.
	With $\beta $ is a natural number denoting the amount of fading, $r$ is a
	parameter defining the detection technique (i.e., $r=1,2$ refers to
	heterodyne detection and IM/DD technique, respectively), while $\mu _{r}$
	denotes the average SNR when $r=1$ and the average electrical SNR when $r=2.$
	Further, $A^{\prime },B^{\prime }$, and $\lambda _{\ell }$ are positive real
	parameters given by%
	\begin{equation*}
	A^{\prime }\triangleq \frac{\left( 1+\frac{\Omega ^{\prime }}{\kappa \beta }%
		\right) ^{1-\beta }}{\Gamma \left( \alpha \right) },\text{ and }B^{\prime
	}\triangleq \frac{\xi ^{2}\alpha \beta \left( \kappa +\Omega ^{\prime
		}\right) }{\left( \xi ^{2}+1\right) \left( \kappa \beta +\Omega ^{\prime
		}\right) },\text{ }\lambda _{\ell }\triangleq \left( _{\ell -1}^{\beta
		-1}\right) \frac{1}{\left( \ell -1\right) !}\left( \frac{\Omega ^{\prime }}{%
		\kappa }\right) ^{\ell -1},
	\end{equation*}
	with $\alpha $ is a positive number related to the effective number of
	large-scale cells of the scattering process, and $\xi $ accounts for the ratio
	between the equivalent beam radius at the receiver to the pointing error
	dis-placement standard deviation. Further, $\kappa =2b_{0}\left( 1-\rho \right) $
	denotes the scattering component's average received power, $2b_{0}$ is the
	average power of the total scatter components, $0\leq \rho \leq 1$
	represents the amount of scattering power coupled to the LOS component, $%
	\Omega ^{\prime }=\Omega +2b_{0}\rho $ denotes the average power from the
	coherent contributions and $\Omega $ is the LOS component's average power.
	Therefore, the closed-form and asymptotic expressions for the three
	considered metrics are summarized in Tables \ref{closedmalaga} and \ref%
	{asympmalaga}, respectively, where $\Xi _{\ell }\triangleq \left\{ \left( \xi ^{2},r\right) ,\left(
	\alpha ,r\right) ,\left( \ell ,r\right) \right\}$, while $\mathcal{L}\left( y\right)$, and $\mathcal{G}\left( y\right)$ are given by (\ref{LLL}) and (\ref{GGG}), respectively.
	\begin{table}[tbh]
		\captionsetup{font=small,format=hang}
		\caption{Closed-form expressions under M\'alaga $\mathcal{M}$ turbulence with
			the presence of pointing errors.}
		\centering\renewcommand{\arraystretch}{1.5}
		\begin{tabular}{c|c|c}
			\hline
			\parbox[t]{2mm}{\multirow{5}{*}{\rotatebox[origin=c]{90}{\centering
						\small{\textit{Ideal RF e2e}}}}} & OP & $1-\sum_{\ell =1}^{\beta }\frac{\psi
				_{\ell }}{\phi _{\ell }}H_{2,4}^{4,0}\left( \phi _{\ell }\gamma _{\text{th}%
			}\left\vert
			\begin{array}{c}
			-;\left( \xi ^{2}+1,r\right) ,\left( 1,1\right)  \\
			\left( 0,1\right) ,\Xi _{\ell };-%
			\end{array}%
			\right. \right) $ \\ \cline{2-3}
			& CC & $\sum_{\ell =1}^{\beta }\frac{\psi _{\ell }}{\log \left( 2\right)
				\phi _{\ell }}H_{3,5}^{5,1}\left( \phi _{\ell }\left\vert
			\begin{array}{c}
			\left( 0,1\right) ;\left( \xi ^{2}+1,r\right) ,\left( 1,1\right)  \\
			\left( 0,1\right) ,\left( 0,1\right) ,\Xi _{\ell };-%
			\end{array}%
			\right. \right) $ \\ \cline{2-3}
			& ASEP & $\sum_{n=0}^{N}\theta _{n}\left[ 1-\sum_{\ell =1}^{\beta }\frac{%
				\psi _{\ell }}{\phi _{\ell }}H_{3,4}^{4,1}\left( \frac{\phi _{\ell }}{\delta
				_{n}}\left\vert
			\begin{array}{c}
			\left( 0,1\right) ;\left( \xi ^{2}+1,r\right) ,\left( 1,1\right)  \\
			\left( 0,1\right) ,\Xi _{\ell };-%
			\end{array}%
			\right. \right) \right] $ \\ \hline\hline
			\parbox[t]{2mm}{\multirow{5}{*}{\rotatebox[origin=c]{90}{\small{%
							\textit{Non-ideal RF e2e}}}}} & OP & $1-\sum_{\ell =1}^{\beta }\frac{\psi
				_{\ell }}{\phi _{\ell }}H_{2,4}^{4,0}\left( \phi _{\ell }\frac{\tau
				_{\varkappa }\gamma _{\text{th}}}{\omega _{\varkappa }-\varrho _{\varkappa
				}\gamma _{\text{th}}}\left\vert
			\begin{array}{c}
			-;\left( \xi ^{2}+1,r\right) ,\left( 1,1\right)  \\
			\left( 0,1\right) ,\Xi _{\ell };-%
			\end{array}%
			\right. \right) $ \\ \cline{2-3}
			& CC & $\frac{1}{\log (2)}\sum_{\ell =1}^{\beta }\frac{\psi _{\ell }}{\phi
				_{\ell }}H_{1,0:1,4:1,2}^{0,1:4,0:1,1}\left( \frac{\phi _{\ell }\tau
				_{\varkappa }}{\varrho _{\varkappa }},\frac{\omega _{\varkappa }}{\varrho
				_{\varkappa }}\left\vert
			\begin{array}{c}
			\left( 1;1,1\right)  \\
			-%
			\end{array}%
			\right\vert
			\begin{array}{c}
			-;\left( \xi ^{2}+1,r\right)  \\
			\left( 0,1\right) ,\Xi _{\ell };-%
			\end{array}%
			\left\vert
			\begin{array}{c}
			\left( 1,1\right) ;- \\
			\left( 1,1\right) ;\left( 0,1\right)
			\end{array}%
			\right. \right) $ \\ \cline{2-3}
			& ASEP & $\sum_{n=0}^{N}\theta _{n}\left[ 1-\sum_{\ell =1}^{\beta }\frac{%
				\psi _{\ell }}{\phi _{\ell }}H_{1,0:1,4:0,2}^{0,1:4,0:1,0}\left( \frac{\phi
				_{\ell }\tau _{\varkappa }}{\varrho _{\varkappa }},\frac{\delta _{n}\omega
				_{\varkappa }}{\varrho _{\varkappa }}\left\vert
			\begin{array}{c}
			\left( 1;1,1\right)  \\
			-%
			\end{array}%
			\right\vert
			\begin{array}{c}
			-;\left( \xi ^{2}+1,r\right)  \\
			\left( 0,1\right) ,\Xi _{\ell };-%
			\end{array}%
			\left\vert
			\begin{array}{c}
			- \\
			\left( 1,1\right) ;\left( 0,1\right)
			\end{array}%
			\right. \right) \right] $ \\ \hline
		\end{tabular}%
		\label{closedmalaga}
	\end{table}
	
	\begin{equation}
	\mathcal{L}\left( y\right) =\frac{\omega _{\varkappa }}{\varrho _{\varkappa
		}\Gamma \left( y\right) }H_{1,1:1,1:1,1}^{0,1:1,1:1,1}\left( \frac{\omega
		_{\varkappa }}{\varrho _{\varkappa }},-1\left\vert
	\begin{array}{c}
	\left( -y;1,1\right) ;- \\
	-;\left( -1-y;1,1\right)
	\end{array}%
	\right\vert
	\begin{array}{c}
	\left( 0,1\right) ;- \\
	\left( 0,1\right) ;-%
	\end{array}%
	\left\vert
	\begin{array}{c}
	\left( 1-y,1\right) ;- \\
	\left( 0,1\right) ;-%
	\end{array}%
	\right. \right),\label{LLL}
	\end{equation}
	\begin{equation}
	\mathcal{G}\left( y\right) =\frac{\omega _{\varkappa }}{\varrho _{\varkappa
		}\Gamma \left( y\right) }H_{1,1:0,1:1,1}^{0,1:1,0:1,1}\left( \frac{\delta
		_{n}\omega _{\varkappa }}{\varrho _{\varkappa }},-1\left\vert
	\begin{array}{c}
	\left( -y;1,1\right) \\
	-;\left( -1-y;1,1\right)%
	\end{array}%
	\right\vert
	\begin{array}{c}
	- \\
	\left( 0,1\right) ;-%
	\end{array}%
	\left\vert
	\begin{array}{c}
	\left( 1-y,1\right) ;- \\
	\left( 0,1\right) ;-%
	\end{array}%
	\right. \right).\label{GGG}
	\end{equation}
	
	\begin{table}[tbh]
		\captionsetup{font=small}
		\caption{Asymptotic expressions under M\'alaga $\mathcal{M}$ turbulence with
			the presence of pointing errors.}
		\label{asympmalaga}\centering\renewcommand{\arraystretch}{1.5}%
		\begin{tabular}{c|c|c}
			\hline
			\parbox[t]{2mm}{\multirow{4}{*}{\rotatebox[origin=c]{90}{\small{\textit{Ideal RF e2e}}}}} & OP & $\sum_{\ell =1}^{\beta }\frac{\psi _{\ell }}{\phi _{\ell }}r%
			\left[
			\begin{array}{c}
			\frac{\Gamma \left( \alpha -\xi ^{2}\right) \Gamma \left( k-\xi ^{2}\right)
			}{\xi ^{2}}\left( \phi _{\ell }\gamma _{\text{th}}\right) ^{\frac{\xi ^{2}}{r%
			}} \\
			+\frac{\Gamma \left( \ell -\alpha \right) }{\alpha \left( \xi ^{2}-\alpha
				\right) }\left( \phi _{\ell }\gamma _{\text{th}}\right) ^{\frac{\alpha }{r}}+%
			\frac{\Gamma \left( \alpha -\ell \right) }{\ell \left( \xi ^{2}-\ell \right)
			}\left( \phi _{\ell }\gamma _{\text{th}}\right) ^{\frac{\ell }{r}}%
			\end{array}%
			\right] $ \\ \cline{2-3}
			& CC & $\sum_{\ell =1}^{\beta }\frac{\psi _{\ell }}{\log \left( 2\right)
				\phi _{\ell }}\frac{\Gamma \left( \alpha \right) \Gamma \left( \ell \right)
			}{\xi ^{2}}\left[ r\Psi \left( \alpha \right) +r\Psi \left( \ell \right) -%
			\frac{r}{\xi ^{2}}-\log \left( \phi _{\ell }\right) \right] $ \\ \cline{2-3}
			& ASEP & $r\sum_{\ell =1}^{\beta }\frac{\psi _{\ell }}{\phi _{\ell }}%
			\sum_{n=0}^{N}\theta _{n}\left[
			\begin{array}{c}
			\frac{\Gamma \left( 1+\frac{\xi ^{2}}{r}\right) \Gamma \left( \alpha -\xi
				^{2}\right) \Gamma \left( \ell -\xi ^{2}\right) }{\xi ^{2}}\left( \frac{\phi
				_{\ell }}{\delta _{n}}\right) ^{\frac{\xi ^{2}}{r}} \\
			+\frac{\Gamma \left( 1+\frac{\alpha }{r}\right) \Gamma \left( \ell -\alpha
				\right) }{\alpha \left( \xi ^{2}-\alpha \right) }\left( \frac{\phi _{\ell }}{%
				\delta _{n}}\right) ^{\frac{\alpha }{r}}+\frac{\Gamma \left( 1+\frac{\ell }{r%
				}\right) \Gamma \left( \alpha -\ell \right) }{\ell \left( \xi ^{2}-\ell
				\right) }\left( \frac{\phi _{\ell }}{\delta _{n}}\right) ^{\frac{\ell }{r}}%
			\end{array}%
			\right] $ \\ \hline\hline
			\parbox[t]{2mm}{\multirow{5}{*}{\rotatebox[origin=c]{90}{\small{\textit{Non-ideal RF e2e}}}}} & OP & $\sum_{\ell =1}^{\beta }r\frac{\psi _{\ell }}{\phi _{\ell }}%
			\left[
			\begin{array}{c}
			\frac{\Gamma \left( \alpha -\xi ^{2}\right) \Gamma \left( \ell -\xi
				^{2}\right) }{\xi ^{2}}\left( \frac{\phi _{\ell }\tau _{\varkappa }}{\varrho
				_{\varkappa }}\right) ^{\frac{\xi ^{2}}{r}}\left( \frac{\gamma _{\text{th}}}{%
				\frac{\omega _{\varkappa }}{\varrho _{\varkappa }}-\gamma _{\text{th}}}%
			\right) ^{\frac{\xi ^{2}}{r}}\text{ } \\
			+\frac{\Gamma \left( \ell -\alpha \right) }{\alpha \left( \xi ^{2}-\alpha
				\right) }\left( \frac{\phi _{\ell }\tau _{\varkappa }}{\varrho _{\varkappa }}%
			\right) ^{\frac{\alpha }{r}}\left( \frac{\gamma _{\text{th}}}{\frac{\omega
					_{\varkappa }}{\varrho _{\varkappa }}-\gamma _{\text{th}}}\right) ^{\frac{%
					\alpha }{r}}+\frac{\Gamma \left( \alpha -\ell \right) }{\ell \left( \xi
				^{2}-\ell \right) }\left( \frac{\phi _{\ell }\tau _{\varkappa }}{\varrho
				_{\varkappa }}\right) ^{\frac{\ell }{r}}\left( \frac{\gamma _{\text{th}}}{%
				\frac{\omega _{\varkappa }}{\varrho _{\varkappa }}-\gamma _{\text{th}}}%
			\right) ^{\frac{\ell }{r}}%
			\end{array}%
			\right] $ \\ \cline{2-3}
			& CC & $\mathcal{C}^{\left( \varkappa ,\infty \right) }-\sum_{\ell
				=1}^{\beta }\frac{\psi _{\ell }}{\log \left( 2\right) \phi _{\ell }}\left[
			\begin{array}{c}
			\frac{\Gamma \left( \alpha -\xi ^{2}\right) \Gamma \left( \ell -\xi
				^{2}\right) }{\xi ^{2}}\left( \frac{\phi _{\ell }\tau _{\varkappa }}{\varrho
				_{\varkappa }}\right) ^{\xi ^{2}}\mathcal{L}\left( \xi ^{2}\right) +\frac{%
				\Gamma \left( \ell -\alpha \right) }{\alpha \left( \xi ^{2}-\alpha \right) }%
			\left( \frac{\phi _{\ell }\tau _{\varkappa }}{\varrho _{\varkappa }}\right)
			^{\alpha }\mathcal{L}\left( \alpha \right) \\
			+\frac{\Gamma \left( \alpha -\ell \right) }{\ell \left( \xi ^{2}-\ell
				\right) }\left( \frac{\phi _{\ell }\tau _{\varkappa }}{\varrho _{\varkappa }}%
			\right) ^{\ell }\mathcal{L}\left( \ell \right)%
			\end{array}%
			\right] $ \\ \cline{2-3}
			& ASEP & $P_{s}^{\left( \varkappa ,\infty \right) }+\sum_{\ell =1}^{\beta }%
			\frac{\psi _{\ell }}{\phi _{\ell }}\sum_{n=0}^{N}\theta _{n}\delta _{n}\left[
			\begin{array}{c}
			\frac{\Gamma \left( \alpha -\xi ^{2}\right) \Gamma \left( \ell -\xi
				^{2}\right) }{\xi ^{2}}\left( \frac{\phi _{\ell }\tau _{\varkappa }}{\varrho
				_{\varkappa }}\right) ^{\xi ^{2}}\mathcal{G}\left( \xi ^{2}\right) \\
			+\frac{\Gamma \left( \ell -\alpha \right) }{\alpha \left( \xi ^{2}-\alpha
				\right) }\left( \frac{\phi _{\ell }\tau _{\varkappa }}{\varrho _{\varkappa }}%
			\right) ^{\alpha }\mathcal{G}\left( \alpha \right) +\frac{\Gamma \left(
				\alpha -\ell \right) }{\ell \left( \xi ^{2}-\ell \right) }\left( \frac{\phi
				_{\ell }\tau _{\varkappa }}{\varrho _{\varkappa }}\right) ^{\ell }\mathcal{G}%
			\left( \ell \right)%
			\end{array}%
			\right] $ \\ \hline
		\end{tabular}%
	\end{table}
	
	\section{Numerical Results}
	
	In this section, we evaluate and illustrate the effects of the RF impairments
	on the performance of a WCS subject to either $\alpha-\mu$ fading
	or M\'alaga $\mathcal{M}$ turbulence channel with the presence of pointing errors under heterodyne technique detection. To this end, the OP, CC, and ASEP are illustrated, assuming $IRR_{\varkappa}=20$ dB, $g_{\varkappa}<1$, $\varphi=3
	{^{\circ}}$, and $\kappa_{t}=\kappa_{r}=0.2$. All the considered scenarios are treated, \textcolor{black}{namely (i) Tx impaired by both IQI and RHI, (ii) Rx impaired by both IQI and RHI, and (iii) both Tx and Rx are impaired.}
	It is noted that in all figures, the \textcolor{black}{analytical} results are
	shown \textcolor{black}{in} continuous lines, the markers are referring to the simulation ones, whereas the dashed lines illustrate the asymptotic curves.
	
	Figures 1 and 2 present
	the OP versus the normalized SNR for \textcolor{black}{all considered scenarios}. \textcolor{black}{It can be seen that} the OP degradation is more severe when the IQI is considered at Tx rather than Rx. \textcolor{black}{This can be explained as the noise is scaled by $\frac{1+g_{\varkappa}^{2}}{2}$
		when Rx impaired only by IQI as can be ascertained in (\ref{yRXonly})} which doesn't exceed $1$ as $g_{\varkappa }<1$. \textcolor{black}{Furthermore, one} can
	notice a minor performance loss caused by RHI impairments. \newline Figures
	3 and 4 illustrate the effects of RF impairments on the \textcolor{black}{CC}. \textcolor{black}{As expected,} the IQI and RF impairments have a detrimental effect, \textcolor{black}{which} becomes more defective when the RHI impairments \textcolor{black}{are}
	added, especially for high SNR \textcolor{black}{values}. Further, it can be
	seen that the presence of the RF impairments \textcolor{black}{leads to a steady of the CC above a certain threshold, as discussed in Remark \ref{remarkCapaceiling}}.\newline  Finally, to quantify the impact of the RF impairments
	on the ASEP, we depict in Figures 5-18 the ASEP for GC-DQPSK, $M$-PSK, $M$-DPSK, and $M$-QAM \textcolor{black}{modulation techniques}. It is first observed that the RF impairments
	exhibit different effects on \textcolor{black}{all the considered modulation schemes}.
	For instance, it can be notice\textcolor{black}{d that the effect of the RF impairments on the ASEP is minor in the case of GC-DQPSK compared to \textcolor{black}{other }$M$-ary modulation techniques}. On the other hand, the effects such impairments \textcolor{black}{on the ASEP} become more detrimental
	with the increasing of the modulation parameter $M$, while for a fixed value
	of $M$, the impact these impairments on $M$-DPSK modulation is \textcolor{black}{more significant} compared \textcolor{black}{to the remaining considered} modulation schemes.
	
	\section{Conclusion}
	
	In this paper, a general framework for
	\textcolor{black}{the analysis of OP,
		CC, and ASEP} of a WCS fading channel in the
	presence of both IQI and RHI at the RF e2e was developed. \textcolor{black}{Precisely}, three realistic %
	\textcolor{black}{possible} cases were \textcolor{black}{then} considered,
	\textcolor{black}{namely (i) Tx impaired by both IQI and RHI, (ii) Rx impaired
		by both IQI and RHI, and (iii) both Tx and Rx are impaired.} The
	corresponding metrics' closed-form, approximate, and asymptotic expressions were evaluated and useful insights into the overall system
	performance we provided. Particularly, new tight approximate expression of the SEP for
	numerous $M$-ary coherent and non-coherent modulation schemes were derived,
	based on which, tight approximate and asymptotic expressions for
	the ASEP were deduced. Numerical simulations confirmed analytical results for all considered impairments and fading models. As a result, we demonstrated that all metrics have ceilings that can't be crossed for any arbitrary fading model and entire SNR range. Moreover, it was shown that the RF impairments have a detrimental impact on the WCS performance and that such an influence becomes severe with the increase of modulation
	\textcolor{black}{parameter $M$. Noticeably, the $M$-DPSK is the most sensitive modulation scheme to the RF impairments. Essentially, the RF impairments must be taken
		into consideration in the efficient design of a WCS as its performance in high
		SNR regime is irrespective of the fading model, while it is significantly
		impacted by the RF impairments parameters.}
	\clearpage
	\begin{figure}[tbp]
		\vspace{-1.2cm}
		\begin{minipage}[t]{0.55\columnwidth}
			\hspace{-1cm}\includegraphics[width=10cm, height=7cm]{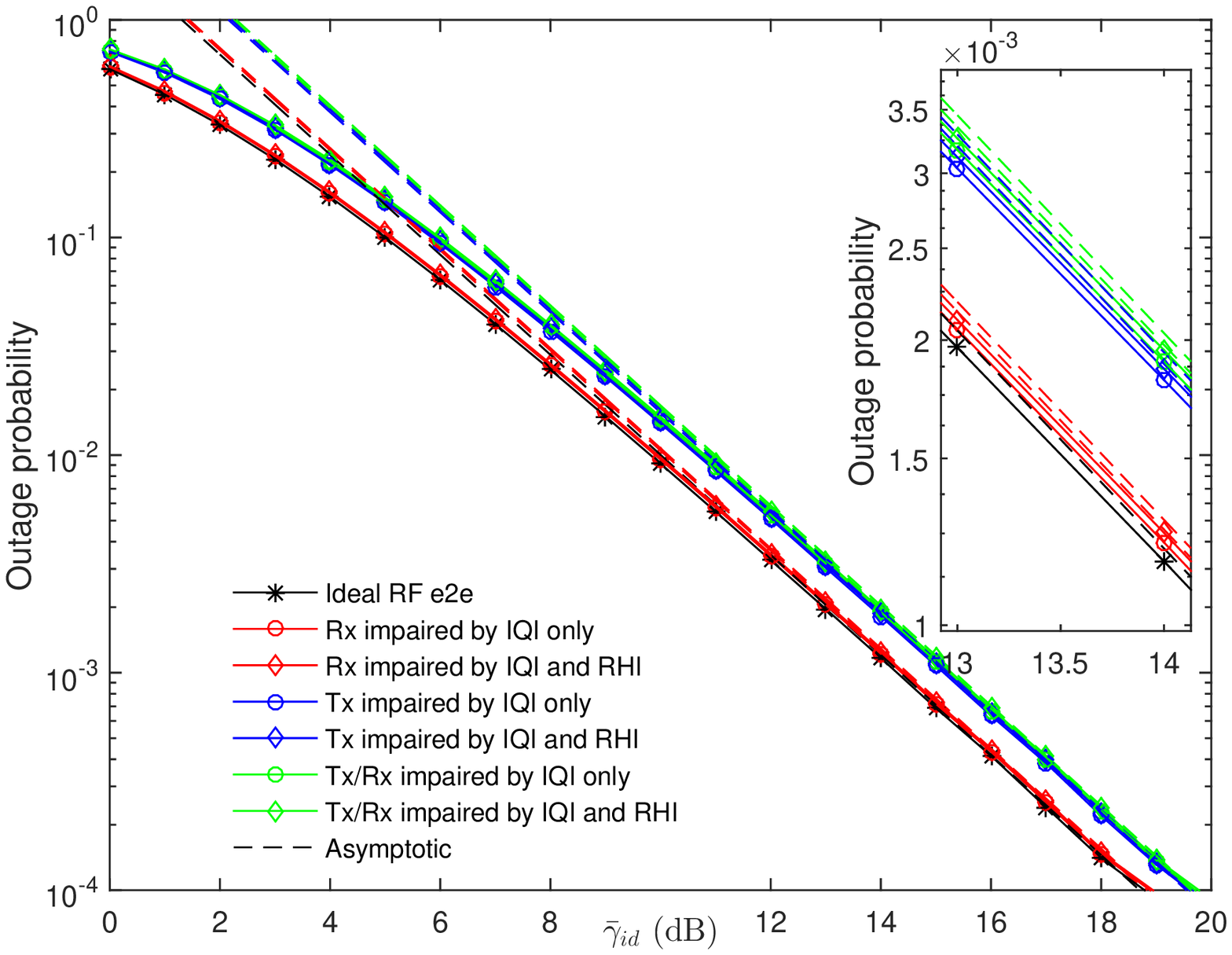}
			\captionsetup{format=plain,font=scriptsize,skip=0pt}
			\caption{$P_{\text{out}}$ versus the normalized SNR over $\alpha -\mu $ fading channel \\with $\alpha=2.3$ and $\mu=2.$ }
			\label{figOpalpha}
		\end{minipage}
		\hfill %
		\hspace{0.2cm}
		\begin{minipage}[t]{0.55\columnwidth}
			\hspace{-1cm}\includegraphics[width=10cm, height=7cm]{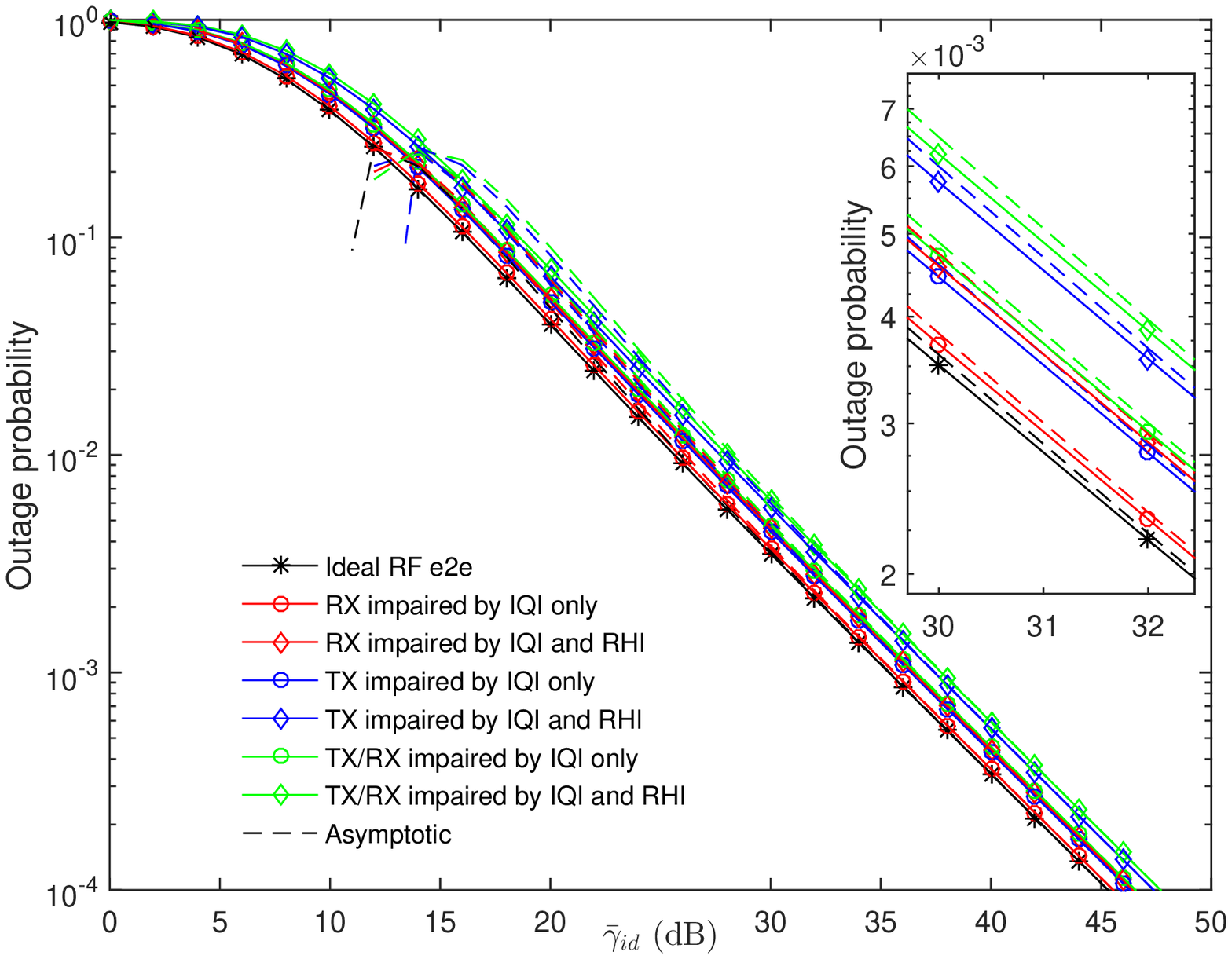}
			\captionsetup{format=plain,font=scriptsize,skip=0pt,margin=-0.2cm}
			\caption {$P_{\text{out}}$ versus the normalized SNR over M\'alaga $\mathcal{M}$ turbulence channel\\ with $\alpha =2.296,\beta =2,\Omega =1,\rho =0.596,\Omega ^{\prime }=1.3265$ and $\xi=3.85$.}
			\label{figOPmalaga}
		\end{minipage}
		\hfill
		\begin{minipage}[t]{0.55\columnwidth}
			
			\hspace{-1cm} \includegraphics[width=10cm, height=7cm]{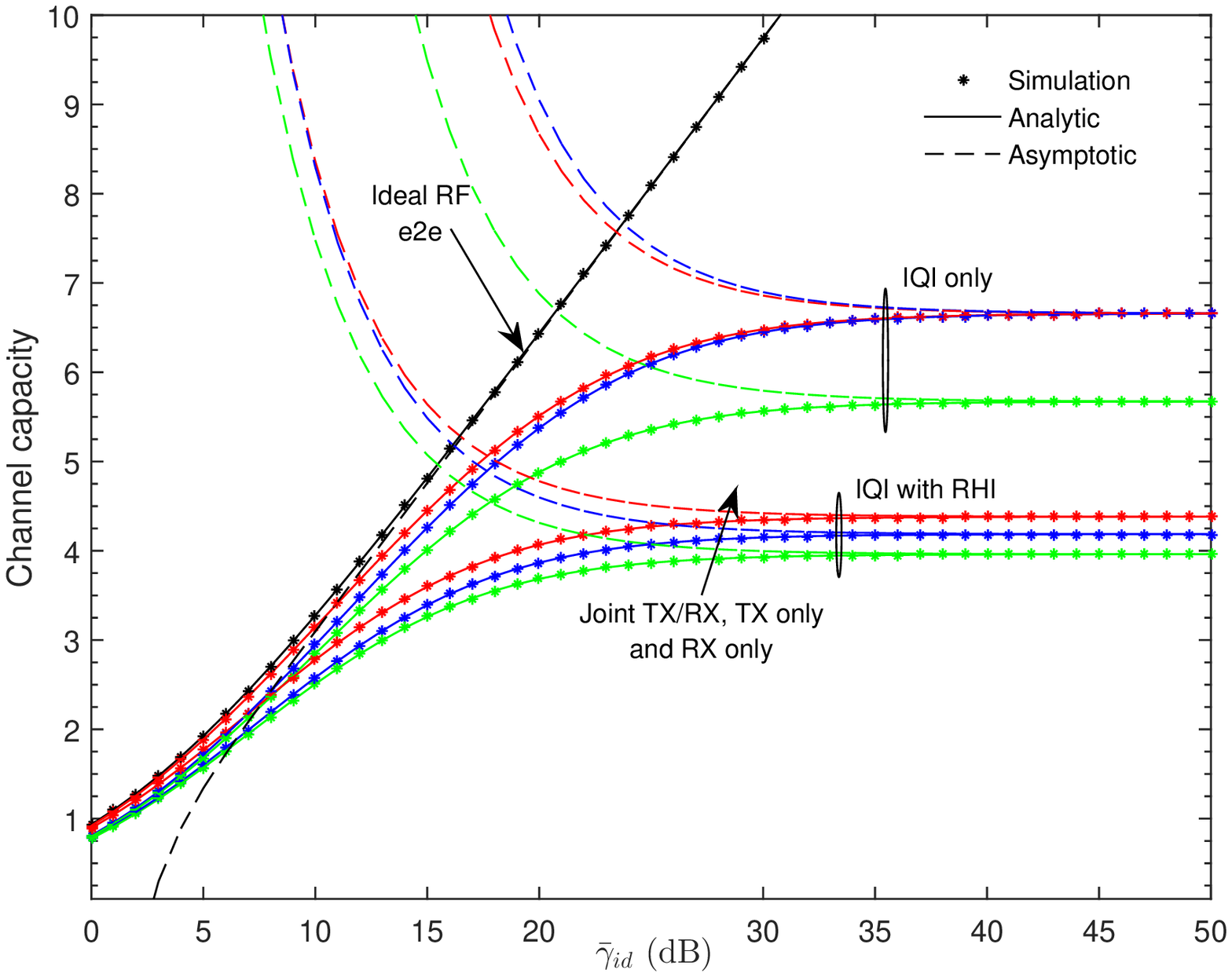}
			\captionsetup{format=plain,font=scriptsize,skip=0pt}
			\caption{Channel capacity versus SNR over $\alpha -\mu $ fading channel \\with $\alpha=3$ and $\mu=2.3 $.}
			\label{figcapalpha}
		\end{minipage}
		\hfill \hspace{0.1cm}
		\begin{minipage}[t]{0.55\columnwidth}
			
			\hspace{-1cm}\includegraphics[width=10cm, height=7cm]{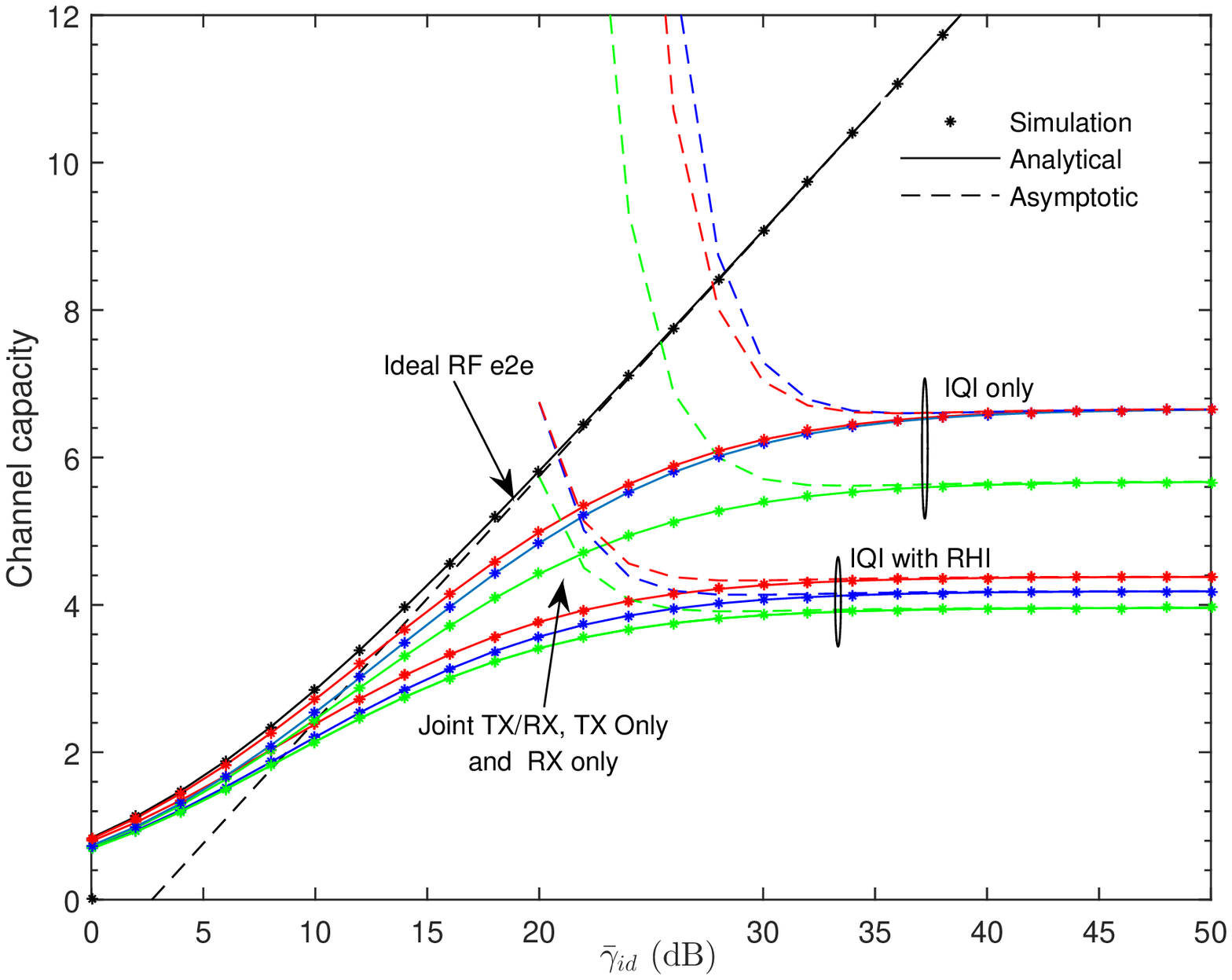}
			\vspace*{-8mm}
			\captionsetup{format=plain,font=scriptsize,skip=0pt,margin=-0.2cm}
			\caption{Channel capacity versus SNR over M\'alaga $\mathcal{M}$ turbulence channel\\ with $\alpha =2.296,\beta =2,\Omega =1,\rho =0.596,\Omega ^{\prime }=1.3265$ and $\xi=3.85$.}
			\label{figcapamalaga}
		\end{minipage}
		\hfill \hspace{0.1cm}
		\begin{minipage}[t]{0.55\columnwidth}
			
			\hspace{-1cm}\includegraphics[width=10cm, height=7cm]{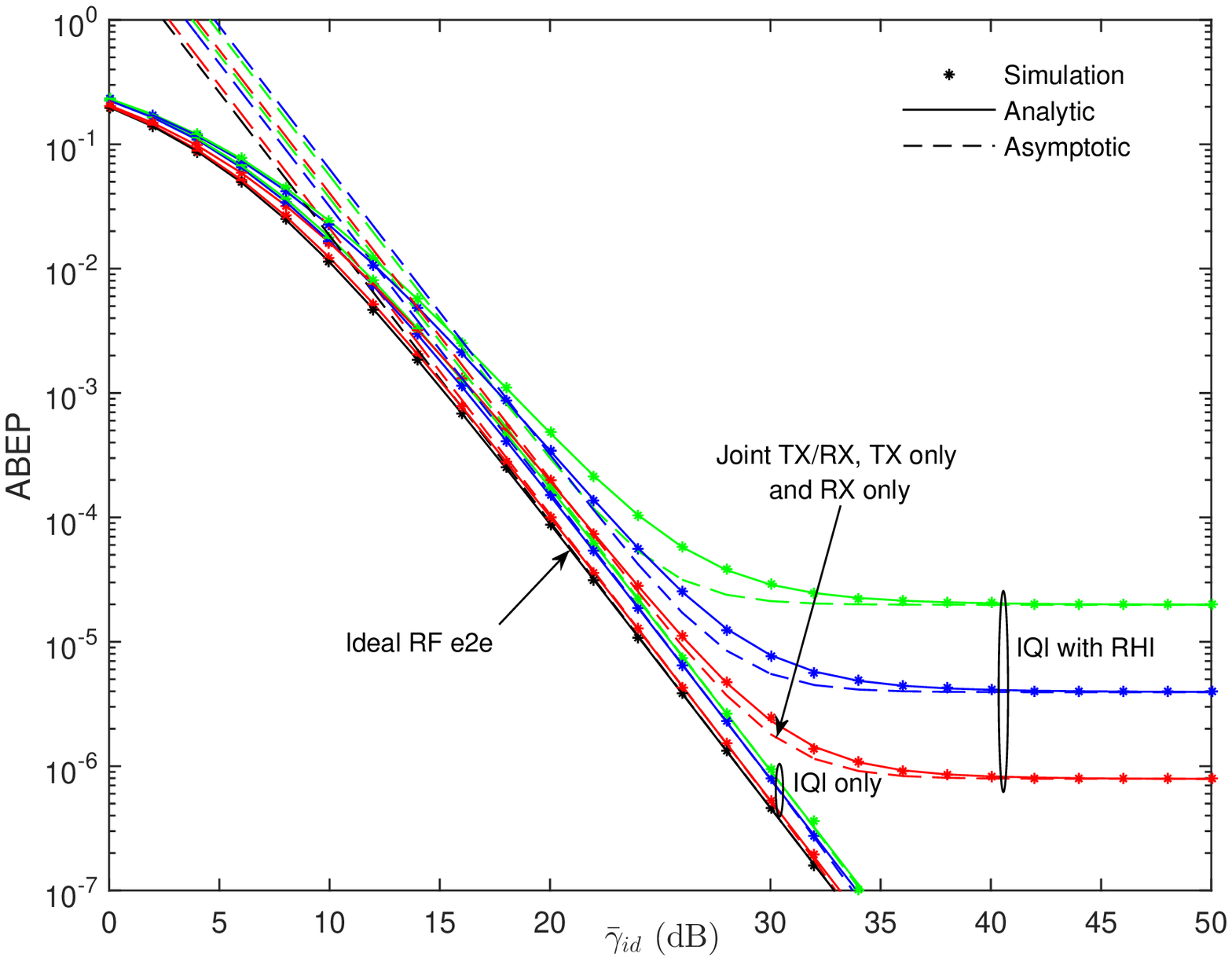}
			\captionsetup{format=plain,font=scriptsize,skip=0pt}
			\caption{ASEP for GC-DQPSK over $\alpha -\mu $ fading channel with\\ $\alpha =2.3$ and $\mu =2$.}
			\label{figdqpskalpha}
		\end{minipage}
		\hfill \hspace{0.1cm}
		\begin{minipage}[t]{0.55\columnwidth}
			
			\hspace{-1cm}\includegraphics[width=10cm, height=7cm]{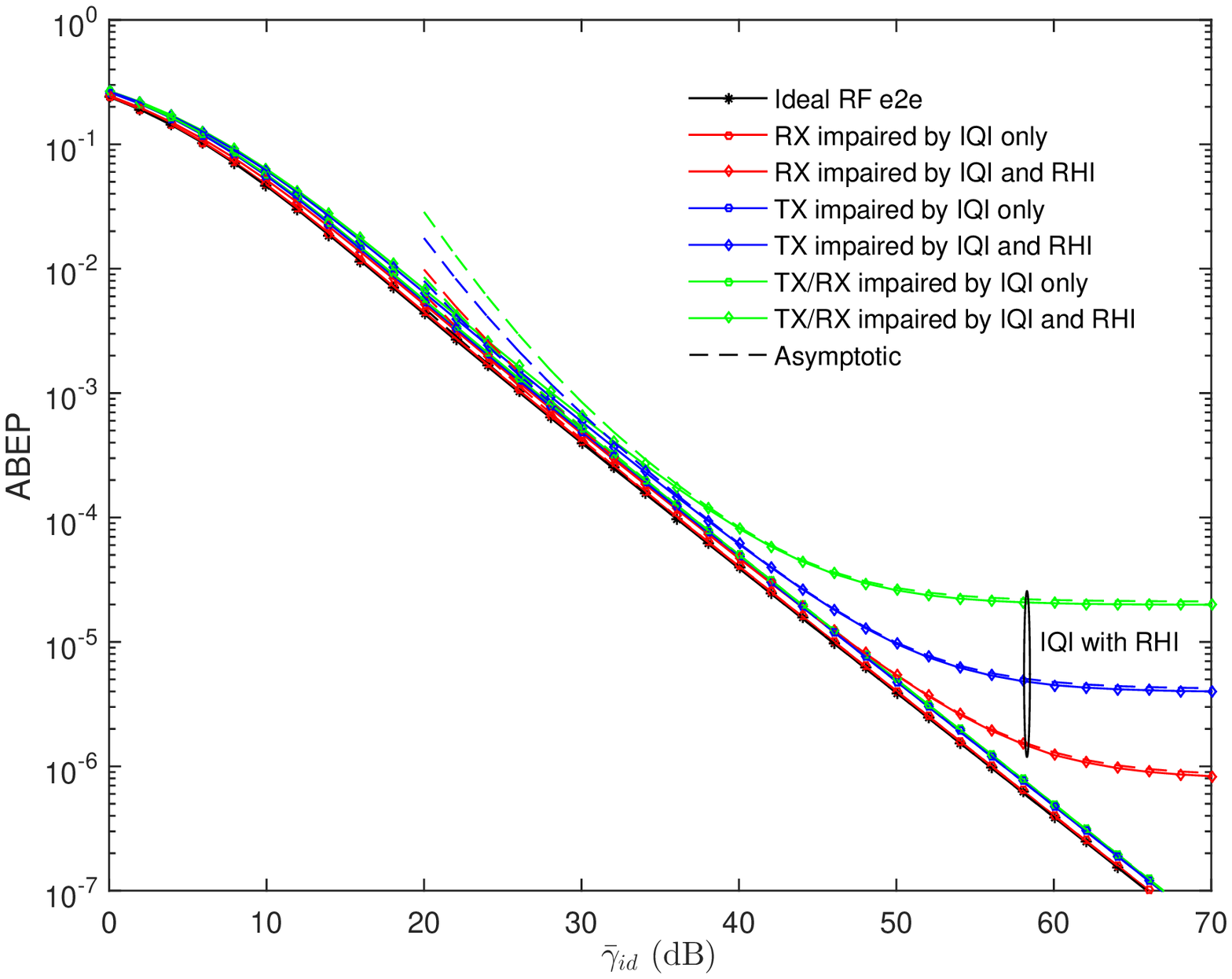}
			\captionsetup{format=plain,font=scriptsize,skip=0pt,margin=-0.2cm}
			\caption{ASEP for GC-DQPSK over M\'alaga $\mathcal{M}$ turbulence channel with\\ $\alpha =2.296,\beta =2,\Omega =1,\rho =0.596,\Omega ^{\prime }=1.3265$ and $\xi=3.85$.}
			\label{figdqpskmalaga}
		\end{minipage}
	\end{figure}
	\newpage
	\begin{figure}[tbp]
		\vspace{-1.2cm}
		\par
		\begin{minipage}[t]{0.55\columnwidth}
			
			\hspace{-1cm}\includegraphics[width=10cm, height=7cm]{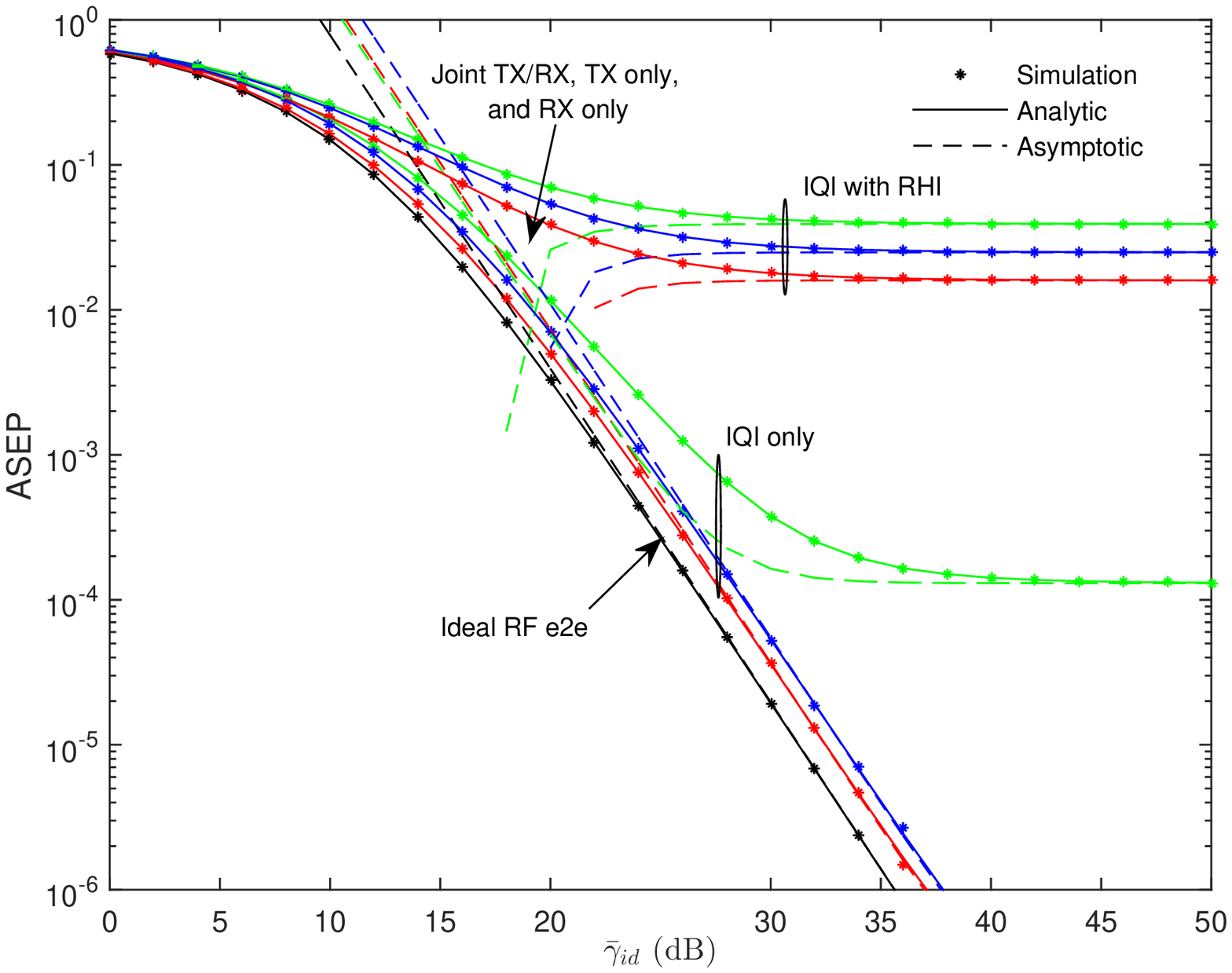}
			\captionsetup{format=plain,font=scriptsize,skip=0pt}
			\caption{ASEP for $8-$PSK over $\alpha -\mu $ fading channel with\\$\alpha =2.3$ and $\mu =2.$}
			\label{fig8mpskalpha}
		\end{minipage}
		\hfill \hspace{0.1cm}
		\begin{minipage}[t]{0.55\columnwidth}
			
			\hspace{-1cm}\includegraphics[width=10cm, height=7cm]{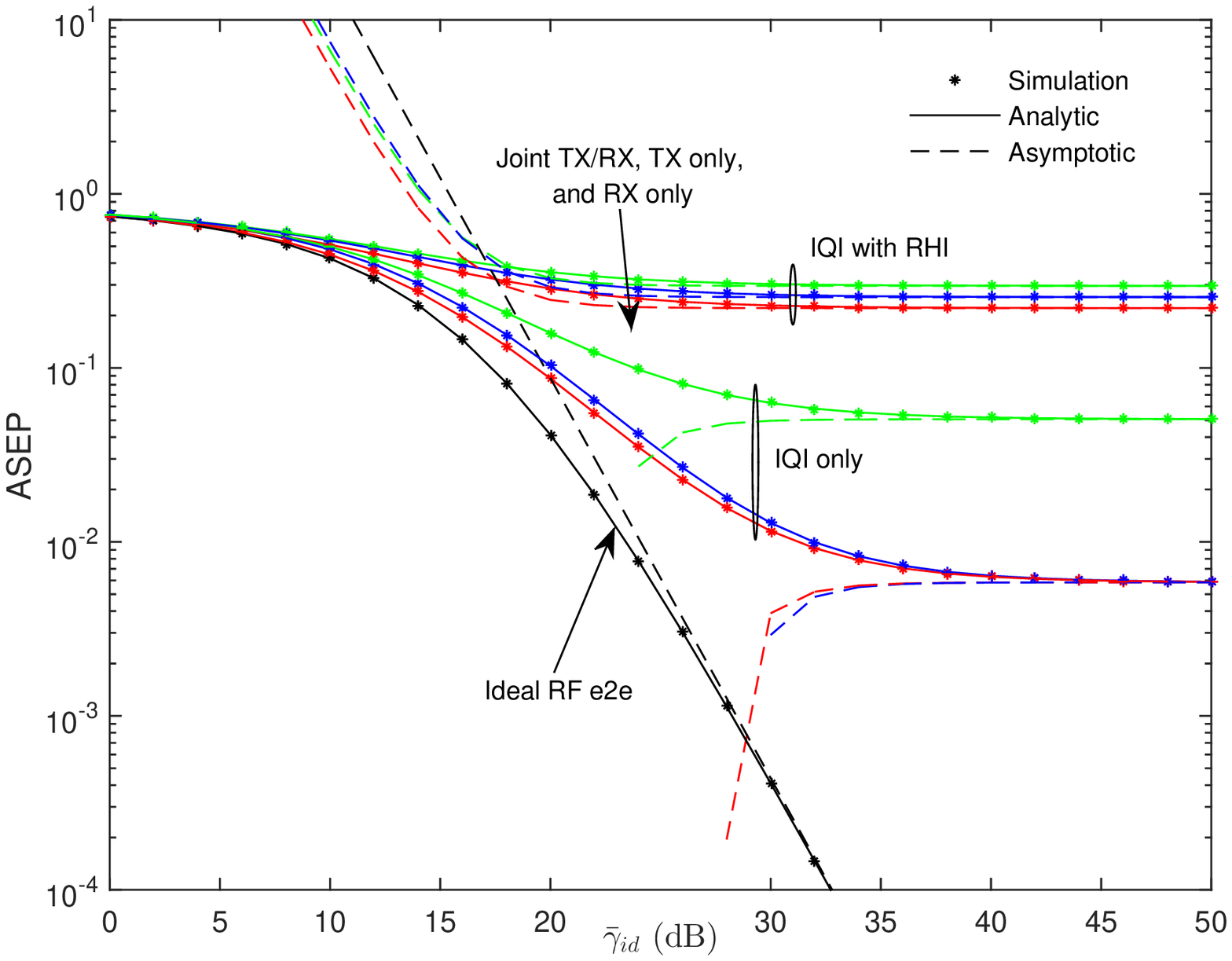}
			\vspace*{-8mm}
			\captionsetup{format=plain,font=scriptsize,skip=0pt}
			\caption{ASEP for $16-$PSK over $\alpha -\mu $ fading channel with\\$\alpha =2.3$ and $\mu =2.$}
			\label{fig16mpskalpha}
		\end{minipage}
		\hfill
		\par
		\begin{minipage}[t]{0.55\columnwidth}
			
			\hspace{-1.5cm} \includegraphics[width=10cm, height=7cm]{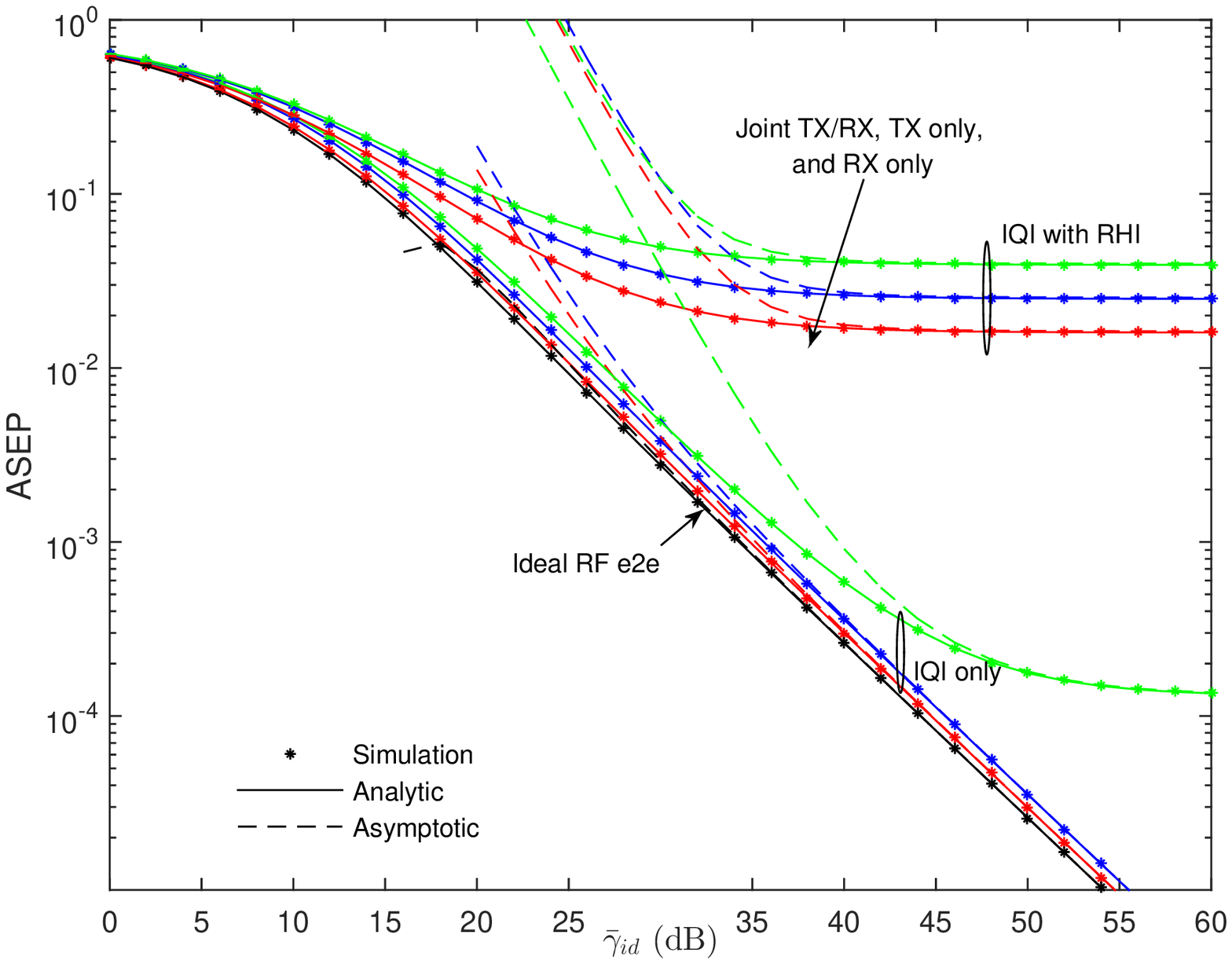}
			\captionsetup{format=plain,font=scriptsize,skip=0pt,margin=-0.2cm}
			\caption{ASEP for $8-$PSK over M\'alaga $\mathcal{M}$ turbulence channel with\\ $\alpha =2.296,\beta =2,\Omega =1,\rho =0.596,\Omega ^{\prime }=1.3265$ and $\xi=3.85$.}
			\label{fig8mpskmalga}
		\end{minipage}
		\hfill \hspace{0.1cm}
		\begin{minipage}[t]{0.55\columnwidth}
			
			\hspace{-1cm}\includegraphics[width=10cm, height=7cm]{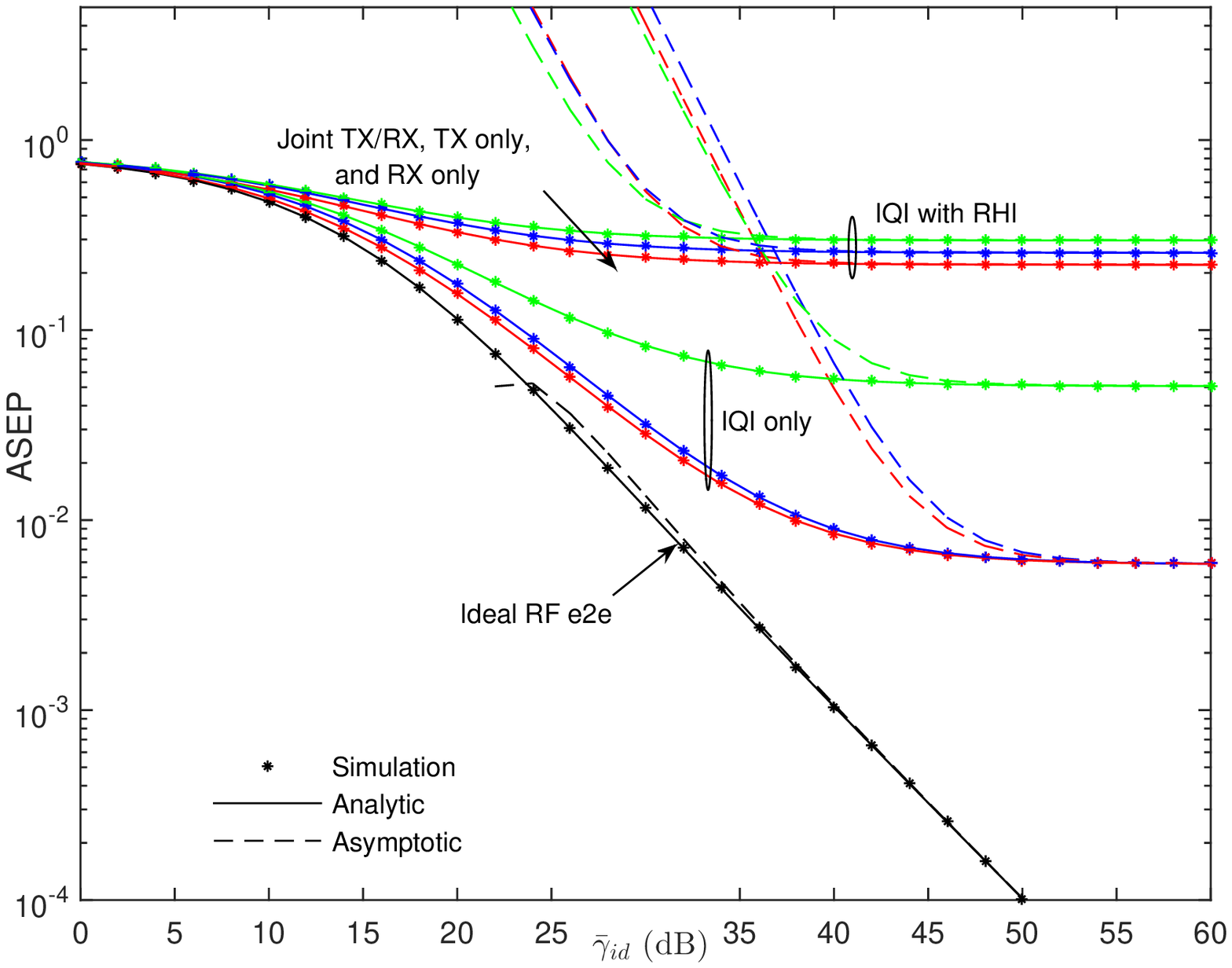}
			\vspace*{-8mm}
			\captionsetup{format=plain,font=scriptsize,skip=0pt,margin=-0.2cm}
			\caption{ASEP for $16-$PSK over M\'alaga $\mathcal{M}$ turbulence channel with\\ $\alpha =2.296,\beta =2,\Omega =1,\rho =0.596,\Omega ^{\prime }=1.3265$ and $\xi=3.85$.}
			\label{fig16mpskmalga}
		\end{minipage}
		\hfill \hspace{0.1cm}
		\begin{minipage}[t]{0.55\columnwidth}
			
			\hspace{-1cm}\includegraphics[width=10cm, height=7cm]{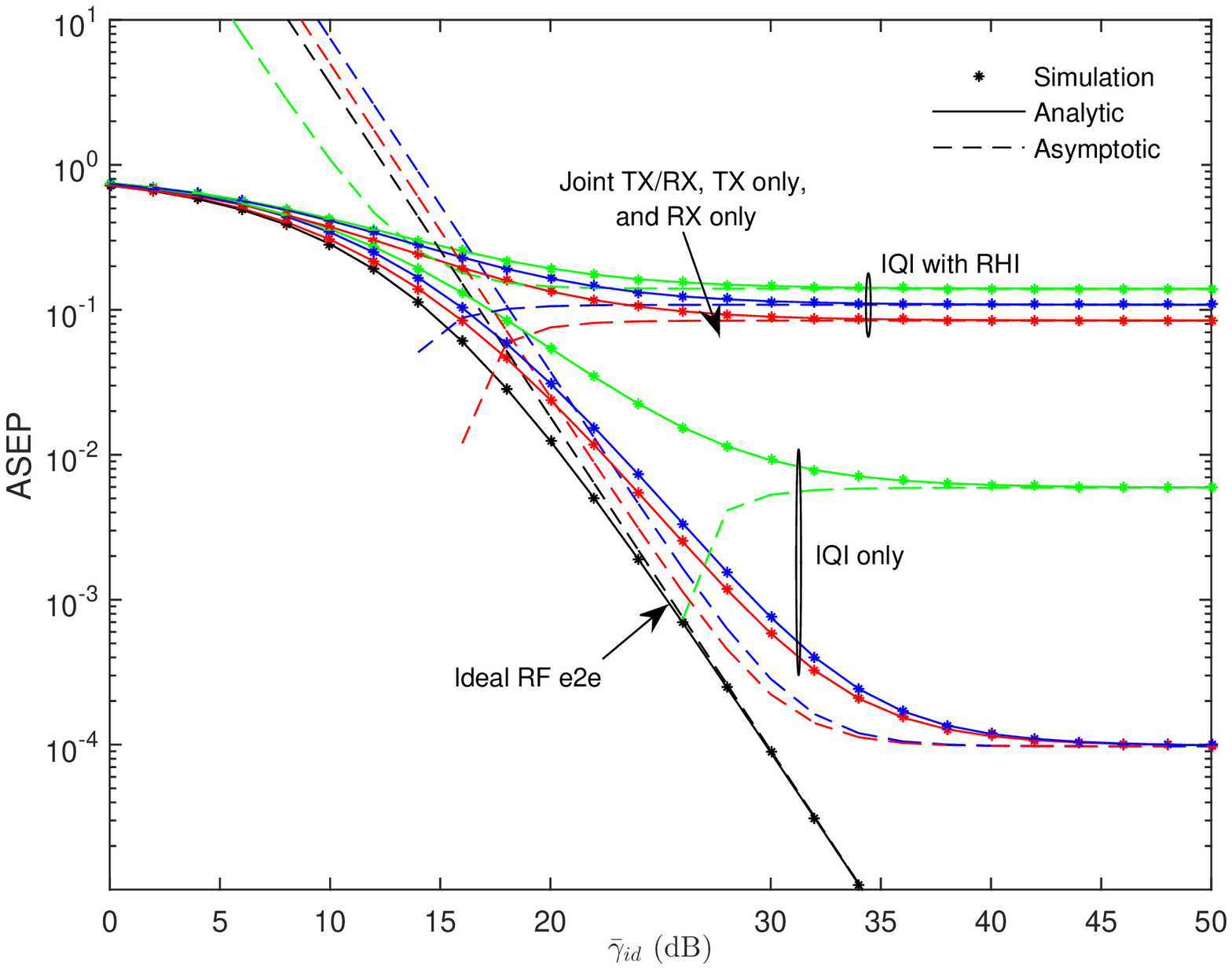}
			\vspace*{-8mm}
			\captionsetup{format=plain,font=scriptsize,skip=0pt}
			\caption{ASEP for $8-$DPSK over $\alpha -\mu $ fading channel with\\$\alpha =2.3$ and $\mu =2.$}
			\label{8mdpskalpha}
		\end{minipage}
		\hfill \hspace{0.1cm}
		\begin{minipage}[t]{0.55\columnwidth}
			
			\hspace{-1cm}\includegraphics[width=10cm, height=7cm]{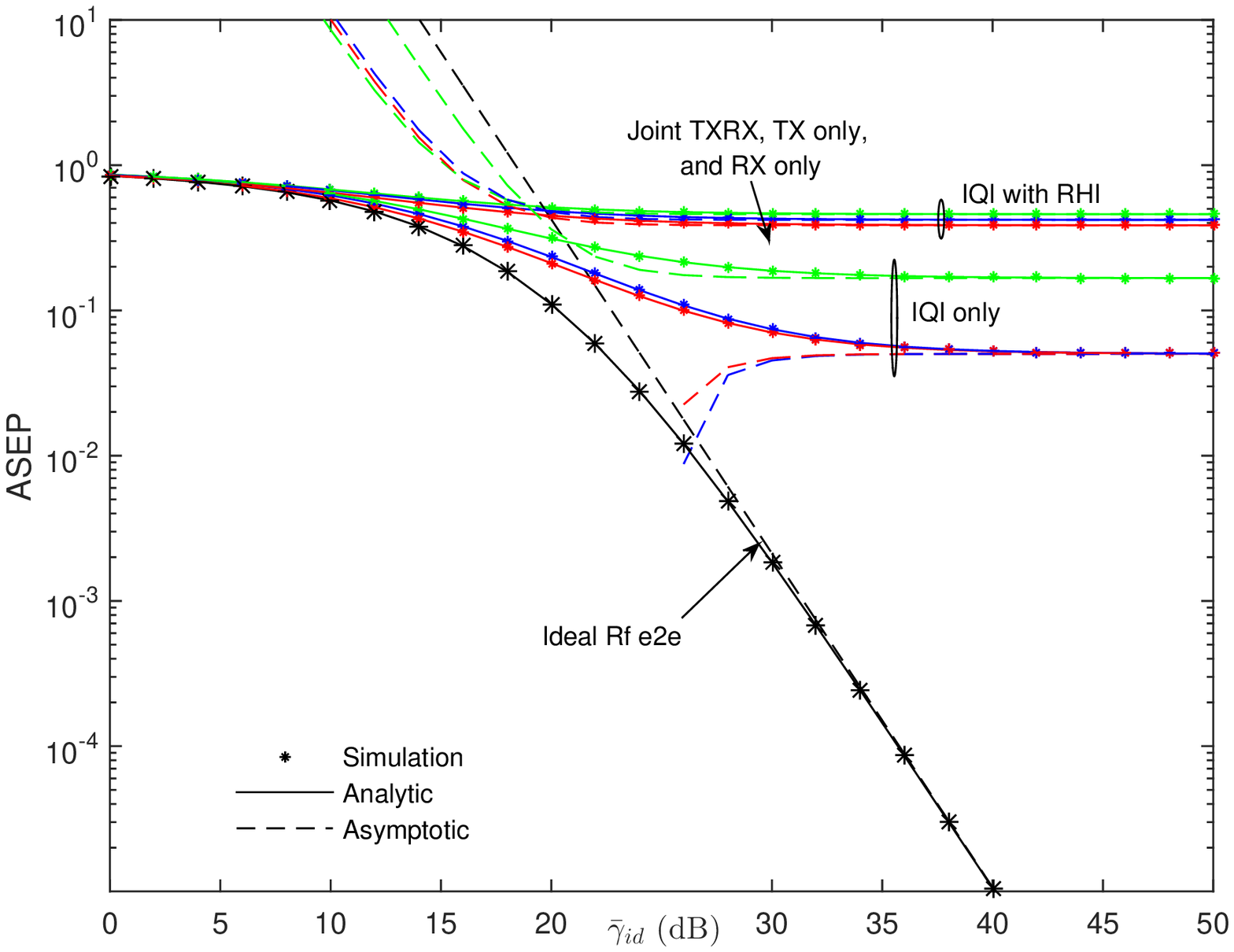}
			\vspace*{-8mm}
			\captionsetup{format=plain,font=scriptsize,skip=0pt}
			\caption{ASEP for $16-$DPSK over $\alpha -\mu $ fading channel with\\$\alpha =2.3$ and $\mu =2.$}
			\label{fig16mdpskalpha}
		\end{minipage}
	\end{figure}
	\newpage
	\begin{figure}[tbp]
		\vspace{-1.2cm}
		\par
		\begin{minipage}[t]{0.55\columnwidth}
			
			\hspace{-1cm}\includegraphics[width=10cm, height=7cm]{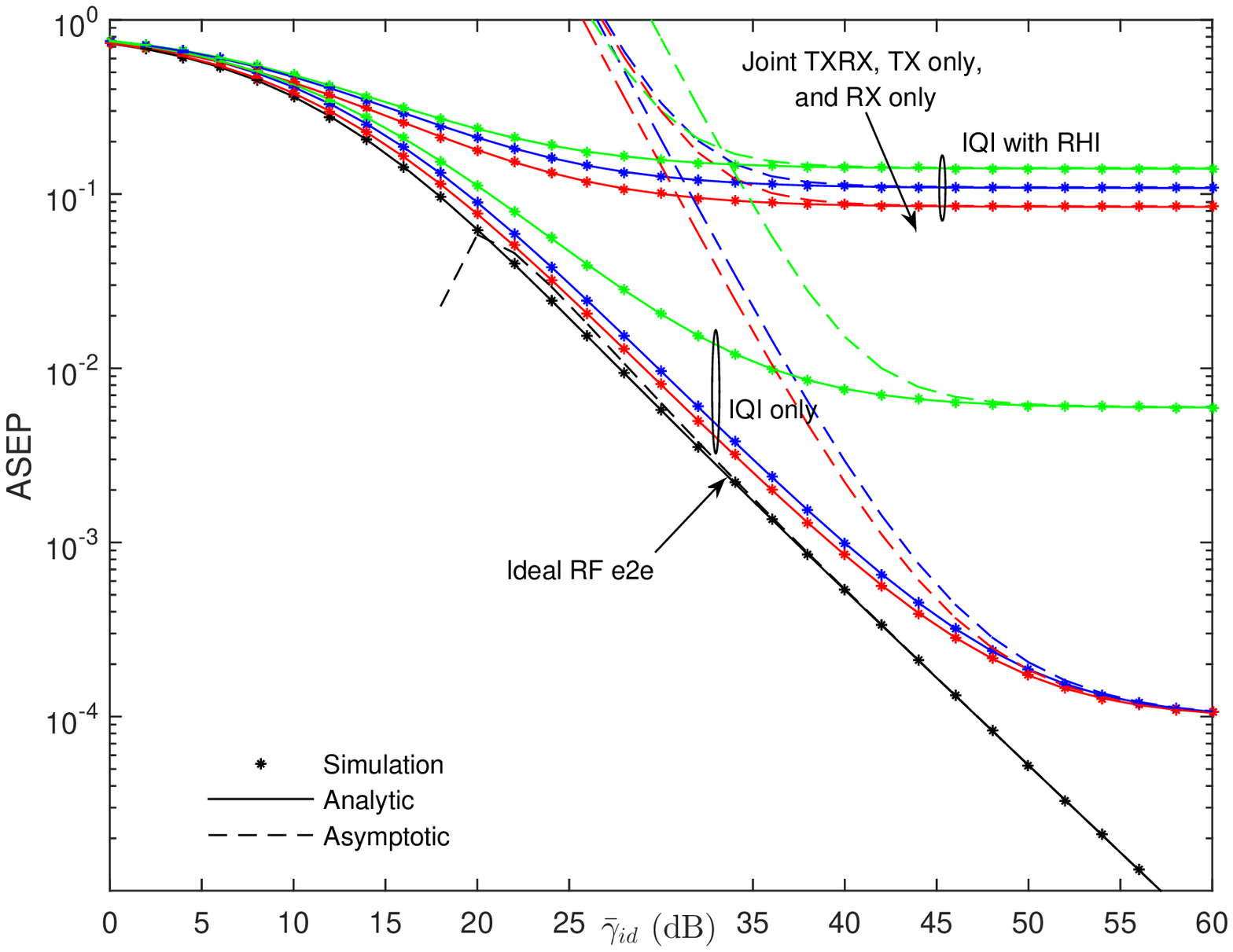}
			\captionsetup{format=plain,font=scriptsize,skip=0pt,margin=-0.2cm}
			\caption{ASEP for $8-$DPSK over M\'alaga $\mathcal{M}$ turbulence channel with\\ $\alpha =2.296,\beta =2,\Omega =1,\rho =0.596,\Omega ^{\prime }=1.3265$ and $\xi=3.85$.}
			\label{fig8mdpskmalaga}
		\end{minipage}
		\hfill \hspace{0.1cm}
		\begin{minipage}[t]{0.55\columnwidth}
			
			\hspace{-1cm}\includegraphics[width=10cm, height=7cm]{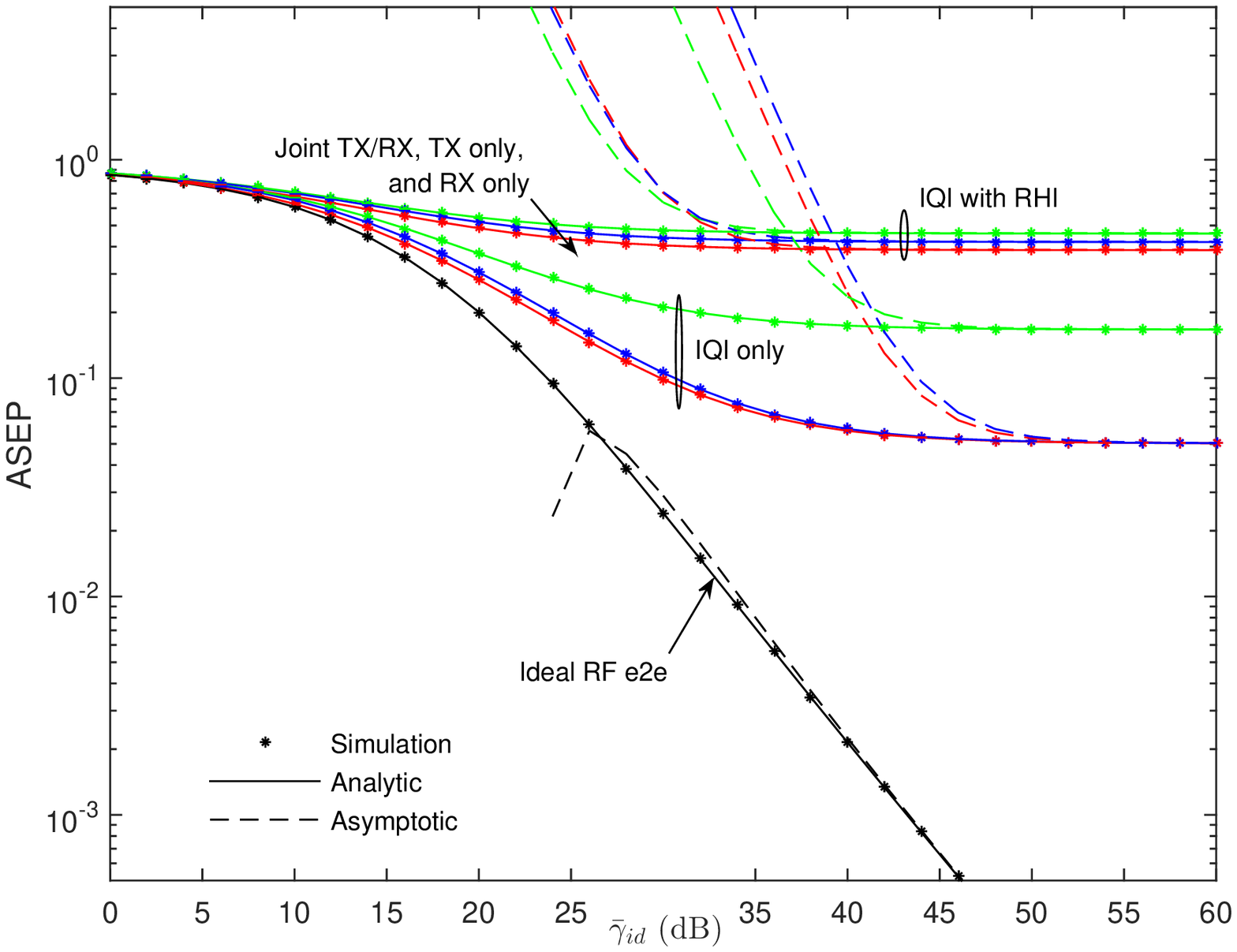}
			\vspace*{-8mm}
			\captionsetup{format=plain,font=scriptsize,skip=0pt,margin=-0.2cm}
			\caption{ASEP for $16-$DPSK over M\'alaga $\mathcal{M}$ turbulence channel with\\ $\alpha =2.296,\beta =2,\Omega =1,\rho =0.596,\Omega ^{\prime }=1.3265$ and $\xi=3.85$.}
			\label{fig16mdpskmalaga}
		\end{minipage}
		\hfill
		\par
		\begin{minipage}[t]{0.55\columnwidth}
			\hspace{-1.2cm} \includegraphics[width=10cm, height=7cm]{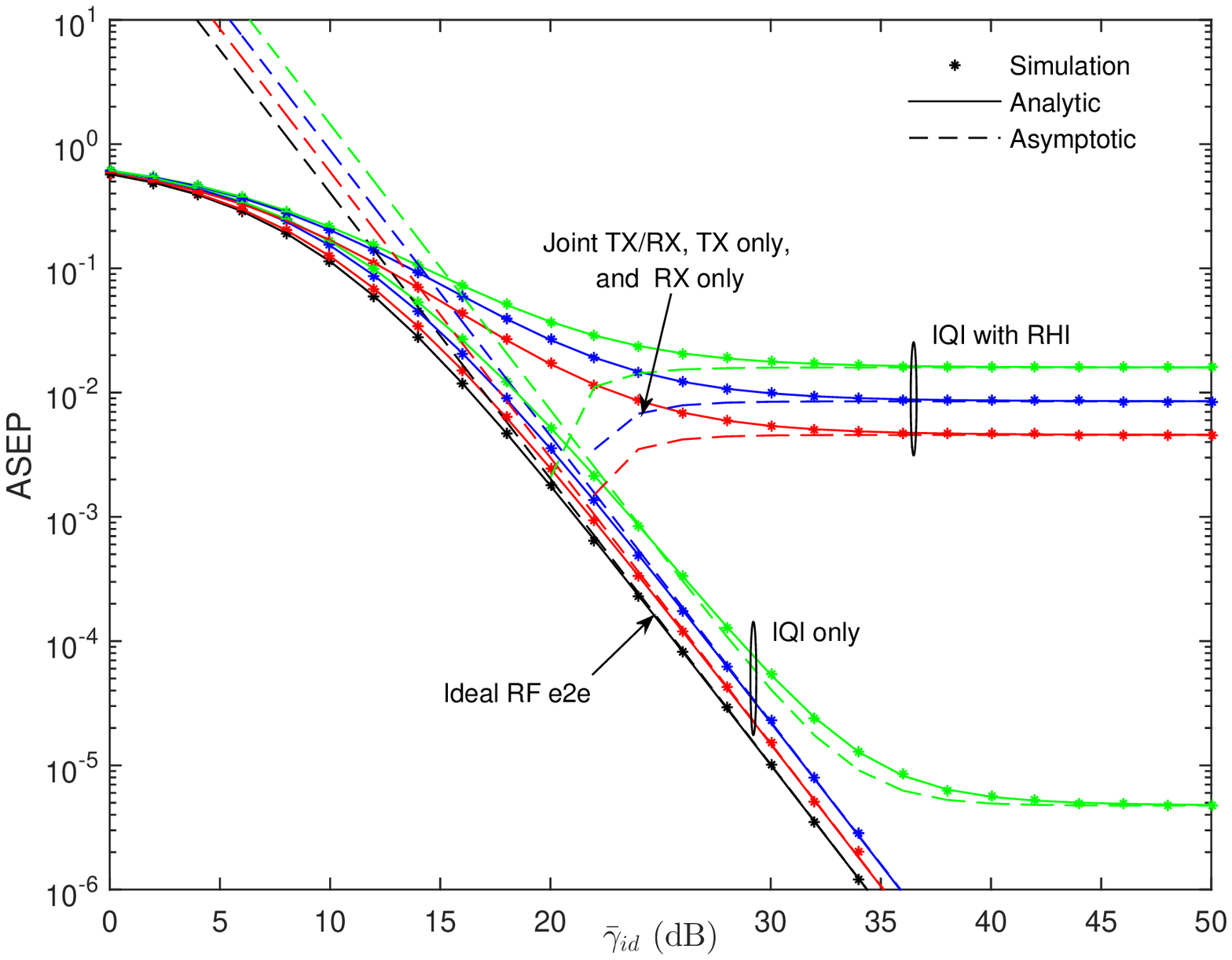}
			\captionsetup{format=plain,font=scriptsize,skip=0pt}
			\caption{ASEP for $8-$QAM over $\alpha -\mu $ fading channel with\\$\alpha =2.3$ and $\mu =2.$}
			\label{fig8qamalpha}
		\end{minipage}
		\hfill \hspace{0.1cm}
		\begin{minipage}[t]{0.55\columnwidth}
			\hspace{-1cm}\includegraphics[width=10cm, height=7cm]{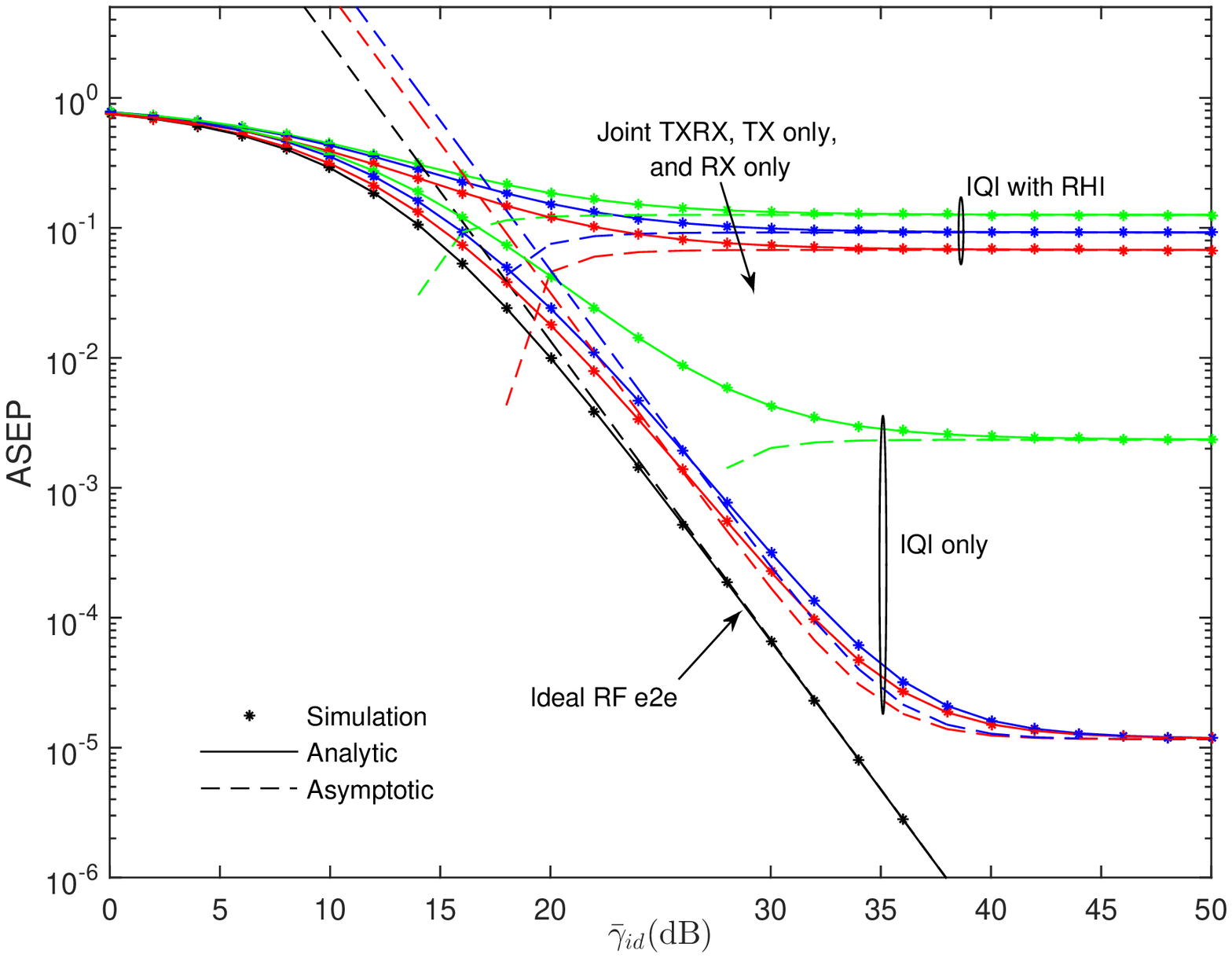}
			\captionsetup{format=plain,font=scriptsize,skip=0pt}
			\caption{ASEP for $16-$QAM over $\alpha -\mu $ fading channel with\\$\alpha =2.3$ and $\mu =2.$}
			\label{fig16qamalpha}
		\end{minipage}
		\hfill \hspace{0.1cm}
		\begin{minipage}[t]{0.55\columnwidth}
			\hspace{-1cm}\includegraphics[width=10cm, height=7cm]{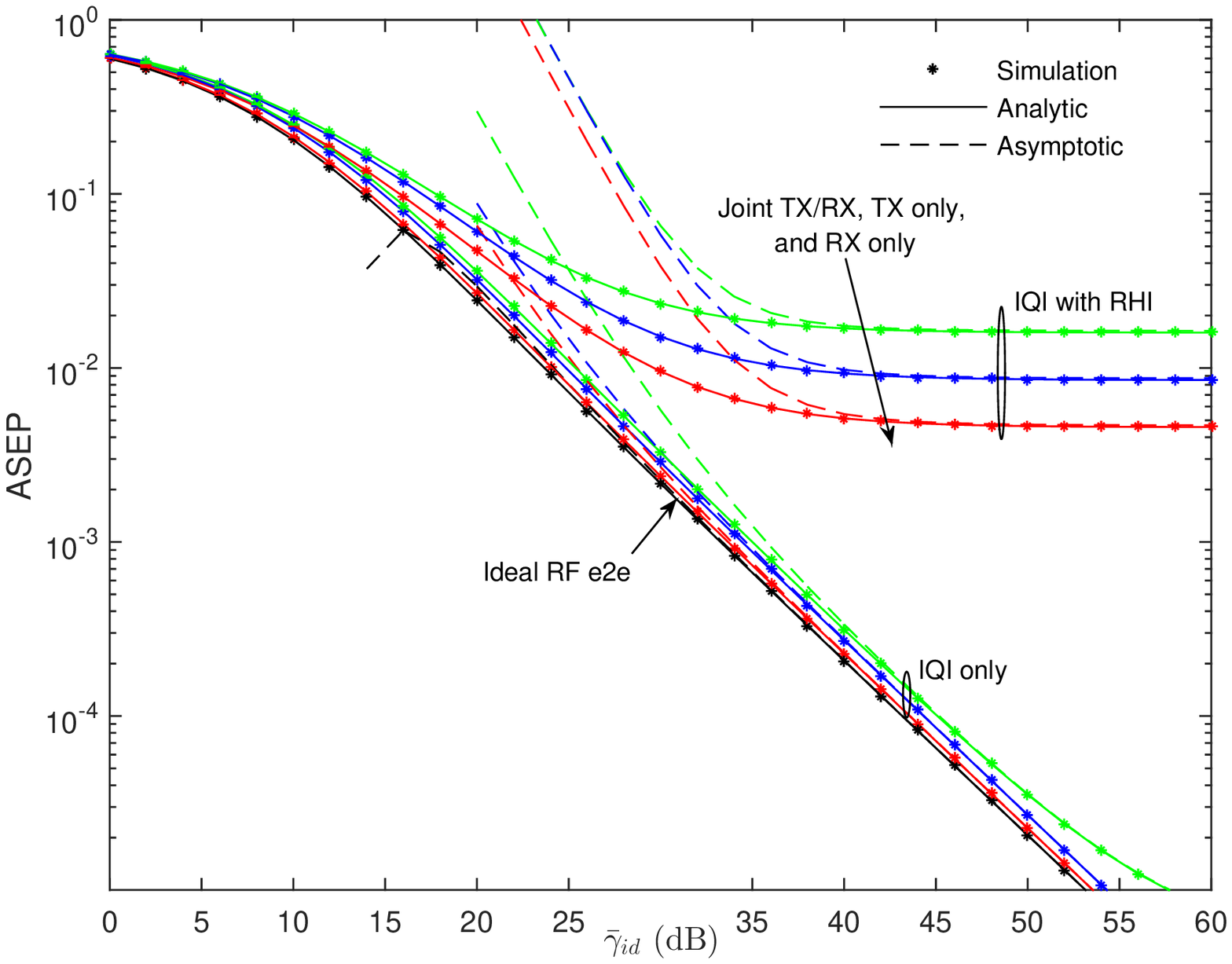}
			\captionsetup{format=plain,font=scriptsize,skip=0pt,margin=-0.2cm}
			\caption{ASEP for $8-$QAM over M\'alaga $\mathcal{M}$ turbulence channel with\\ $\alpha =2.296,\beta =2,\Omega =1,\rho =0.596,\Omega ^{\prime }=1.3265$ and $\xi=3.85$.}
			\label{fig8qammalaga}
		\end{minipage}
		\hfill \hspace{0.1cm}
		\begin{minipage}[t]{0.55\columnwidth}
			\hspace{-1cm}\includegraphics[width=10cm, height=7cm]{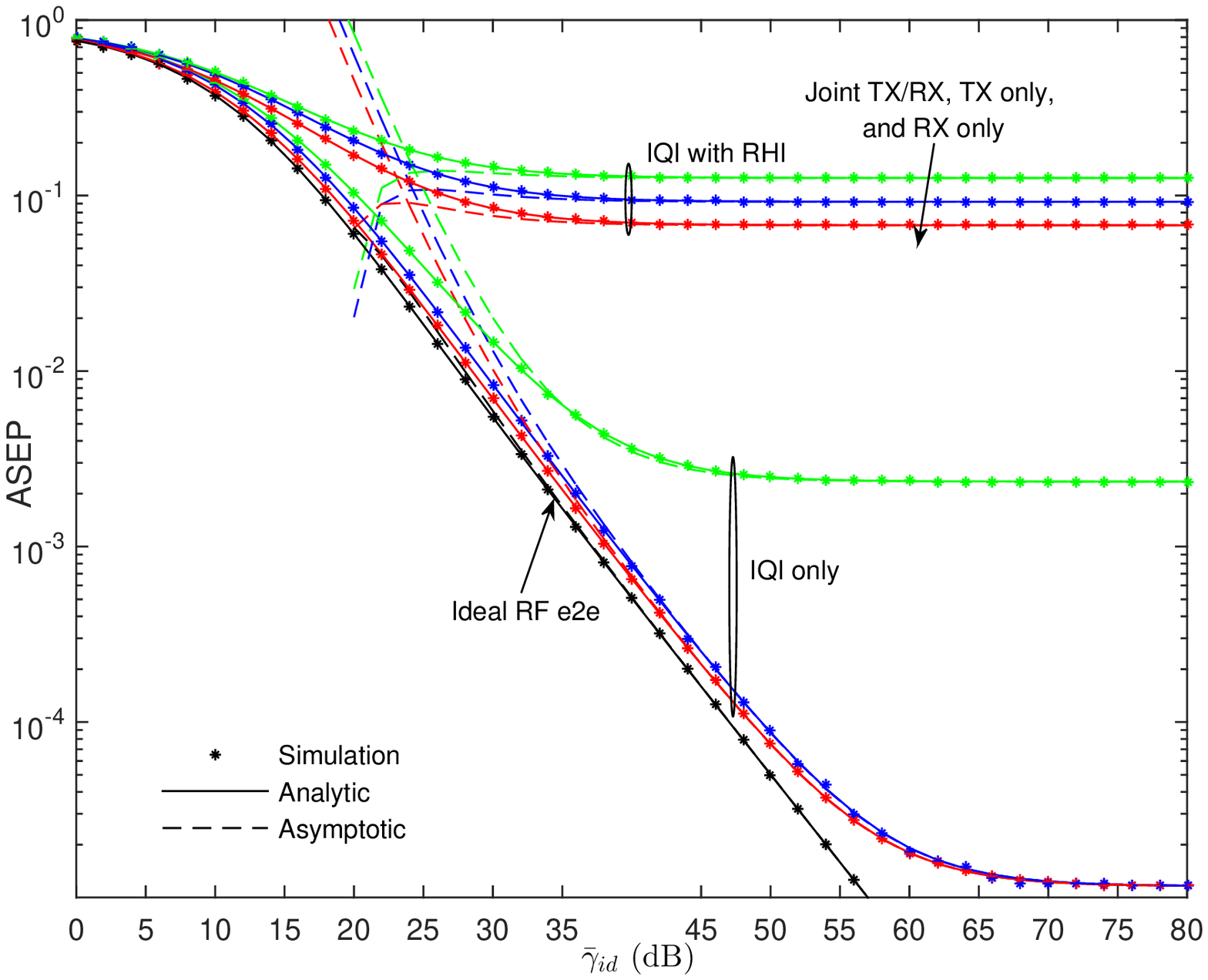}
			\vspace*{-8mm}
			\captionsetup{format=plain,font=scriptsize,skip=0pt,margin=-0.2cm}
			\caption{ASEP for $16-$QAM over M\'alaga $\mathcal{M}$ turbulence channel with\\$\alpha =2.296,\beta =2,\Omega =1,\rho =0.596,\Omega ^{\prime }=1.3265$ and $\xi=3.85$.}
			\label{fig16qammalaga}
		\end{minipage}
	\end{figure}
	\clearpage
	\newpage
	\mbox{}
	\section*{Appendix A: Proof of proposition \protect\ref{ORAexact}}
	\begin{flushleft}
		By setting $U^{\prime }=\frac{\partial F_{\gamma _{\varkappa }}\left( \gamma
			\right) }{\partial \gamma },$ $V=\log (1+\gamma )$, and using the
		integration by parts in $\left( \ref{ORAexpression}\right) $ with the help
		of $\left( \ref{CDFIQI}\right) $, one obtains
		\begin{eqnarray}
		\mathcal{C}^{\left( \varkappa \right) } &=&\frac{1}{\log \left( 2\right) }%
		\sum_{\ell =1}^{L}\frac{\psi _{\ell }}{\phi _{\ell }}\int_{0}^{\frac{\omega
				_{\varkappa }}{\varrho _{\varkappa }}}\frac{1}{1+\gamma }{\small H}_{p+1{\small ,q+1}}^{m+1{\small ,n}}\left( \frac{\phi
			_{\ell }\tau _{\varkappa }\gamma }{\omega _{\varkappa }-\varrho _{\varkappa
			}\gamma }\left\vert
		\begin{array}{c}
		\Psi ^{\left( \ell \right) } \\
		\left( 0,1\right) ,\Upsilon ^{\left( \ell \right) }%
		\end{array}%
		\right. \right) d\gamma.  \label{EQORAappedix}
		\end{eqnarray}
		by plugging (\ref{Hffox}) into (\ref{EQORAappedix}), we get
		\begin{eqnarray}
		\mathcal{C}^{\left( \varkappa \right) } &=&\frac{1}{\log \left( 2\right) }%
		\frac{1}{2\pi j}\sum_{\ell =1}^{L}\frac{\psi _{\ell }}{\phi _{\ell }}\int_{%
			\mathcal{L}_{v}^{\left( \ell \right) }}\frac{\prod\limits_{i=1}^{m}\Gamma
			\left( \mathcal{B}_{i}^{\left( \ell \right) }+B_{i}^{\left( \ell \right)
			}v\right) \prod\limits_{i=1}^{n}\Gamma \left( 1-\mathcal{A}_{i}^{\left( \ell
				\right) }-A_{i}^{\left( \ell \right) }v\right) }{v\prod\limits_{i=n+1}^{p}%
			\Gamma \left( \mathcal{A}_{i}^{\left( \ell \right) }+A_{i}^{\left( \ell
				\right) }v\right) \prod\limits_{i=m+1}^{q}\Gamma \left( 1-\mathcal{B}%
			_{i}^{\left( \ell \right) }-B_{i}^{\left( \ell \right) }v\right) }  \notag \\
		&&\times \left( \frac{\phi _{\ell }\tau _{\varkappa }}{\varrho _{\varkappa }}%
		\right) ^{-v}\underbrace{\int_{0}^{\frac{\omega _{\varkappa }}{\varrho
					_{\varkappa }}}\gamma ^{-v}\left( 1+\gamma \right) ^{-1}\left( \frac{\omega
				_{\varkappa }}{\varrho _{\varkappa }}-\gamma \right) ^{v}d\gamma }_{\mathcal{%
				I}_{\varkappa }}dv.  \label{ORAdetail}
		\end{eqnarray}%
		Further, the inner integral $\mathcal{I}_{\varkappa }$ can be evaluated
		using \cite[Eq.(3.197.8)]{integraltable} as
		\begin{equation}
		\mathcal{I}_{\varkappa }\mathcal{=}\frac{\omega _{\varkappa }}{\varrho
			_{\varkappa }}\Gamma \left( 1+v\right) \Gamma \left( 1-v\right) \text{ }%
		_{2}F_{1}\left( 1,1-v;2;-\frac{\omega _{\varkappa }}{\varrho _{\varkappa }}%
		\right) ,  \label{innerintegral}
		\end{equation}%
		where $_{2}F_{1}\left( .,.;.;.\right) $ denotes the Gauss hypergeometric
		function (GHF) \cite[Eq. (9.100)]{integraltable}.\newline
		Now plugging $\left( \ref{innerintegral}\right) $ into $\left( \ref%
		{ORAdetail}\right) $ and using the integral Mellin-Barnes representation of
		the GHF \cite[Eq. (9.113)]{integraltable}, yields
		\begin{eqnarray}
		\mathcal{C}^{\left( \varkappa \right) } &=&\frac{1}{\log (2)}\left( \frac{1}{%
			2\pi j}\right) ^{2}\sum_{\ell =1}^{L}\frac{\psi _{\ell }}{\phi _{\ell }}%
		\int_{\mathcal{L}_{v}^{\left( \ell \right) }}\int_{\mathcal{L}_{w}^{\left(
				\ell \right) }}\frac{\Gamma \left( {-v-w}\right) \Gamma \left( v\right)
			\prod\limits_{i=1}^{m}\Gamma \left( \mathcal{B}_{i}^{\left( \ell \right)
			}+B_{i}^{\left( \ell \right) }v\right) \prod\limits_{i=1}^{n}\Gamma \left( 1-%
			\mathcal{A}_{i}^{\left( \ell \right) }-A_{i}^{\left( \ell \right) }v\right)
		}{\prod\limits_{i=n+1}^{p}\Gamma \left( \mathcal{A}_{i}^{\left( \ell \right)
			}+A_{i}^{\left( \ell \right) }v\right) \prod\limits_{i=m+1}^{q}\Gamma \left(
			1-\mathcal{B}_{i}^{\left( \ell \right) }-A_{i}^{\left( \ell \right)
			}v\right) }  \notag \\
		&&\times \frac{\Gamma \left( 1+w\right) \Gamma \left( {-w}\right) }{\Gamma
			\left( 1-w\right) }\left( \frac{\phi _{\ell }\tau _{\varkappa }}{\varrho
			_{\varkappa }}\right) ^{-v}\left( \frac{\omega _{\varkappa }}{\varrho
			_{\varkappa }}\right) ^{-w}dvdw,  \label{ORRA}
		\end{eqnarray}%
		which concludes the proof of Proposition 1.
	\end{flushleft}
	
	\section*{Appendix B: Proof of corollary \protect\ref{crolaryidealcapa}}
	
	\begin{flushleft}
		By replacing $\omega _{\varkappa }=\tau _{\varkappa }=1$ in $\left( \ref%
		{ORRA}\right) $, the CC under ORA policy can be simplified to%
		\begin{eqnarray}
		\mathcal{C}^{\left( \varkappa \right) } &=&\frac{1}{\log (2)}\left( \frac{1}{%
			2\pi j}\right) ^{2}\sum_{\ell =1}^{L}\frac{\psi _{\ell }}{\phi _{\ell }}%
		\int_{\mathcal{L}_{v}^{\left( \ell \right) }}\int_{\mathcal{L}_{w}^{\left(
				\ell \right) }}\frac{\Gamma \left( -v-w\right) \Gamma \left( v\right)
			\prod\limits_{i=1}^{m}\Gamma \left( \mathcal{B}_{i}^{\left( \ell \right)
			}+B_{i}^{\left( \ell \right) }v\right) }{\prod\limits_{i=n+1}^{p}\Gamma
			\left( \mathcal{A}_{i}^{\left( \ell \right) }+A_{i}^{\left( \ell \right)
			}v\right) }  \notag \\
		&&\times \frac{\prod\limits_{i=1}^{n}\Gamma \left( 1-\mathcal{A}_{i}^{\left(
				\ell \right) }-A_{i}^{\left( \ell \right) }v\right) }{\prod%
			\limits_{i=m+1}^{q}\Gamma \left( 1-\mathcal{B}_{i}^{\left( \ell \right)
			}-B_{i}^{\left( \ell \right) }v\right) }\frac{\Gamma \left( -w\right) \Gamma
			\left( 1+w\right) }{\Gamma \left( 1-w\right) }\phi _{\ell }^{-v}\left( \frac{%
			1}{\varrho _{\varkappa }}\right) ^{-w-v}dvdw.  \notag
		\end{eqnarray}%
		Using the change of variable $s=$ $-w-v,$ we get
		\begin{eqnarray}
		\mathcal{C}^{\left( \varkappa \right) } &=&\frac{1}{\log (2)}\frac{1}{2\pi j}%
		\sum_{\ell =1}^{L}\frac{\psi _{\ell }}{\phi _{\ell }}\int_{\mathcal{L}%
			_{v}^{\left( \ell \right) }}\frac{\Gamma \left( v\right)
			\prod\limits_{i=1}^{m}\Gamma \left( \mathcal{B}_{i}^{\left( \ell \right)
			}+B_{i}^{\left( \ell \right) }v\right) \prod\limits_{i=1}^{n}\Gamma \left( 1-%
			\mathcal{A}_{i}^{\left( \ell \right) }-A_{i}^{\left( \ell \right) }v\right)
		}{\prod\limits_{i=n+1}^{p}\Gamma \left( \mathcal{A}_{i}^{\left( \ell \right)
			}+A_{i}^{\left( \ell \right) }v\right) \prod\limits_{i=m+1}^{q}\Gamma \left(
			1-\mathcal{B}_{i}^{\left( \ell \right) }-B_{i}^{\left( \ell \right)
			}v\right) }\phi _{\ell }^{-v}  \notag \\
		&&\times \underbrace{\frac{1}{2\pi j}\int_{\mathcal{L}_{s}^{\left( \ell
					\right) }}\frac{\Gamma \left( s+v\right) \Gamma \left( 1-s-v\right) }{\Gamma
				\left( 1+s+v\right) }\Gamma \left( s\right) \varrho _{\varkappa }^{-s}ds}_{%
			\mathcal{K}_{\varkappa }}dv
		\end{eqnarray}%
		For significantly smaller values of $\varrho _{\varkappa }$ (i.e., $\varrho
		_{\varkappa }<<1$), the condition \cite[Eq. (1.2.15)]{kilbas} is satisfied.
		It follows that the inner integral $\mathcal{K}_{\varkappa }$ can be
		expressed as an infinite summation of the residues evaluated at the left poles.
		Subsequently, relying on \cite[Theorem 1.3]{kilbas}, $\mathcal{K}_{\varkappa
		}$ can be expressed as
		\begin{equation}
		\mathcal{K}_{\varkappa }\mathcal{=}\sum_{l=0}^{\infty }\frac{\left(
			-1\right) ^{l}\Gamma \left( -l+v\right) \Gamma \left( 1+l-v\right) }{%
			l!\Gamma \left( 1-l+v\right) }\varrho _{\varkappa }^{l}.  \label{K-residu}
		\end{equation}%
		Particularly, for $\varrho _{\varkappa }=0$ (i.e., case of ideal RF e$2$e), $%
		\left( \ref{K-residu}\right) $ is reduced to
		\begin{equation}
		\mathcal{K}_{\varkappa }\mathcal{=}\frac{\Gamma \left( v\right) \Gamma
			\left( 1-v\right) }{\Gamma \left( 1+v\right) }.  \label{K-res22}
		\end{equation}%
		Thus, the CC under ORA policy and ideal RF e$2$e can be expressed as%
		\begin{eqnarray}
		\mathcal{C}^{\left( \varkappa \right) } &=&\frac{1}{\log (2)}\sum_{\ell
			=1}^{L}\frac{\psi _{\ell }}{\phi _{\ell }}\left( \frac{1}{2\pi j}\right)
		\int_{\mathcal{L}_{v}^{\left( \ell \right) }}\frac{\Gamma \left( v\right)
			\Gamma \left( v\right) \Gamma \left( 1-v\right) }{\Gamma \left( 1+v\right) }
		\notag \\
		&&\times \frac{\prod\limits_{i=1}^{m}\Gamma \left( \mathcal{B}_{i}^{\left(
				\ell \right) }+B_{i}^{\left( \ell \right) }v\right)
			\prod\limits_{i=1}^{n}\Gamma \left( 1-\mathcal{A}_{i}^{\left( \ell \right)
			}-A_{i}^{\left( \ell \right) }v\right) }{\prod\limits_{i=n+1}^{p}\Gamma
			\left( \mathcal{A}_{i}^{\left( \ell \right) }+A_{i}^{\left( \ell \right)
			}v\right) \prod\limits_{i=m+1}^{q}\Gamma \left( 1-\mathcal{B}_{i}^{\left(
				\ell \right) }-B_{i}^{\left( \ell \right) }v\right) }\phi _{\ell }^{-v}dv.
		\end{eqnarray}%
		This concludes the proof.
	\end{flushleft}
	
	\section*{Appendix C: Proof of proposition \protect\ref{propASEPexact}}
	
	\begin{flushleft}
		\textcolor{black}{By substituting $\left( \ref{SEPbin}\right)$ into the
			following equation}
		\begin{equation}
		P_{s}^{\left( \varkappa \right) }=\int_{0}^{\frac{\omega _{\varkappa }}{%
				\varrho _{\varkappa }}}\mathcal{H}\left( \gamma \right) \frac{\partial
			F_{\gamma _{\varkappa }}\left( \gamma \right) }{\partial \gamma }d\gamma ,
		\label{PSexpression}
		\end{equation}%
		one \textcolor{black}{gets a tight approximate formula}
		\begin{equation}
		P_{s}^{\left( \varkappa \right) }\simeq \sum_{n=0}^{N}\theta _{n}\mathcal{I}%
		_{n}^{\left( \varkappa \right) },  \label{Psaprox}
		\end{equation}%
		with
		\begin{equation}
		\mathcal{I}_{n}^{\left( \varkappa \right) }=\int_{0}^{\frac{\omega
				_{\varkappa }}{\varrho _{\varkappa }}}\exp \left( -\delta _{n}\gamma \right)
		\frac{\partial F_{\gamma _{\varkappa }}\left( \gamma \right) }{\partial
			\gamma }d\gamma .  \label{integral_n}
		\end{equation}%
		By setting $U^{\prime }=\frac{\partial F_{\gamma _{\varkappa }}\left( \gamma
			\right) }{\partial \gamma }$, $V=\exp \left( -\delta _{n}\gamma \right) $,
		and using $\left( \ref{CDFIQI}\right) $, \textcolor{black}{one can obtain}
		\begin{eqnarray}
		\mathcal{I}_{n}^{\left( \varkappa \right) } &=&1-\delta _{n}\sum_{\ell
			=1}^{L}\frac{\psi _{\ell }}{\phi _{\ell }}\int_{0}^{\frac{\omega _{\varkappa
			}}{\varrho _{\varkappa }}}\exp \left( -\delta _{n}\gamma \right)   \notag \\
		&&\times {\small H}_{p+1{\small ,q+1}}^{m+1{\small ,n}}\left( \frac{\phi
			_{\ell }\tau _{\varkappa }\gamma }{\omega _{\varkappa }-\varrho _{\varkappa
			}\gamma }\left\vert
		\begin{array}{c}
		\left( \mathcal{A}_{i}^{\left( \ell \right) },A_{i}^{\left( \ell \right)
		}\right) _{i=1:p},\left( 1,1\right)  \\
		\left( 0,1\right) ,\left( \mathcal{B}_{i}^{\left( \ell \right)
		},B_{i}^{\left( \ell \right) }\right) _{i=1:q_{{}}}%
		\end{array}%
		\right. \right) d\gamma .  \notag \\
		&&  \label{EqIntegral-n}
		\end{eqnarray}%
		\textcolor{black}{Subsequently,} $\left( \ref{EqIntegral-n}\right) $ can be
		rewritten with the help of $\left( \ref{Hffox}\right) $ as
		\begin{eqnarray}
		\mathcal{I}_{n}^{\left( \varkappa \right) } &=&1-\frac{\delta _{n}}{2\pi j}%
		\sum_{\ell =1}^{L}\frac{\psi _{\ell }}{\phi _{\ell }}\int_{\mathcal{L}%
			_{v}^{\left( \ell \right) }}\frac{\prod\limits_{i=1}^{m}\Gamma \left(
			\mathcal{B}_{i}^{\left( \ell \right) }+B_{i}^{\left( \ell \right) }v\right)
			\prod\limits_{i=1}^{n}\Gamma \left( 1-\mathcal{A}_{i}^{\left( \ell \right)
			}-A_{i}^{\left( \ell \right) }v\right) }{v\prod\limits_{i=n+1}^{p}\Gamma
			\left( \mathcal{A}_{i}^{\left( \ell \right) }+A_{i}^{\left( \ell \right)
			}v\right) \prod\limits_{i=m+1}^{q}\Gamma \left( 1-\mathcal{B}_{i}^{\left(
				\ell \right) }-B_{i}^{\left( \ell \right) }v\right) }  \notag \\
		&&\times \left( \frac{\phi _{\ell }\tau _{\varkappa }}{\varrho _{\varkappa }}%
		\right) ^{-v}\underbrace{\int_{0}^{\frac{\omega _{\varkappa }}{\varrho
					_{\varkappa }}}\gamma ^{-v}\left( \frac{\omega _{\varkappa }}{\varrho
				_{\varkappa }}-\gamma \right) ^{v}\exp \left( -\delta _{n}\gamma \right)
			d\gamma }_{\mathcal{J}_{\varkappa }}dv.  \label{integ2}
		\end{eqnarray}%
		\textcolor{black}{whereas} the inner integral $\mathcal{J}_{\varkappa }$ can
		be evaluated using \cite[Eq. (3.383.1)]{integraltable} as
		\begin{equation}
		\mathcal{J}_{\varkappa }\mathcal{=}\Gamma \left( 1+v\right) \frac{\omega
			_{\varkappa }}{\varrho _{\varkappa }}\text{ }_{1}F_{1}\left( 1-v;2;-\frac{%
			\delta _{n}\omega _{\varkappa }}{\varrho _{\varkappa }}\right) ,
		\label{J_integral}
		\end{equation}%
		where $_{1}F_{1}\left( .;.;.\right) $ denotes the confluent hypergeometric function \cite[ Eq. (9.21)]{integraltable}. \textcolor{black}{Lastly,}
		plugging $\left( \ref{J_integral}\right) $ into $\left( \ref{integ2}\right) $
		and \textcolor{black}{using} \cite[Eqs. (2.9.14)]{kilbas} %
		\textcolor{black}{along} with (\ref{Hffox}), yields
		\begin{eqnarray}
		\mathcal{I}_{n}^{\left( \varkappa \right) } &=&1-\left( \frac{1}{2\pi j}%
		\right) ^{2}\sum_{\ell =1}^{L}\frac{\psi _{\ell }}{\phi _{\ell }}\int_{%
			\mathcal{L}_{v}^{\left( \ell \right) }}\int_{\mathcal{L}_{s}^{\left( \ell
				\right) }}\Gamma \left( -v-s\right) \frac{\Gamma \left( v\right)
			\prod\limits_{i=1}^{m}\Gamma \left( \mathcal{B}_{i}^{\left( \ell \right)
			}+B_{i}^{\left( \ell \right) }v\right) }{\prod\limits_{i=n+1}^{p}\Gamma
			\left( \mathcal{A}_{i}^{\left( \ell \right) }+A_{i}^{\left( \ell \right)
			}v\right) }  \notag \\
		&&\times \frac{\prod\limits_{i=1}^{n}\Gamma \left( 1-\mathcal{A}_{i}^{\left(
				\ell \right) }-A_{i}^{\left( \ell \right) }v\right) }{\prod%
			\limits_{i=m+1}^{q}\Gamma \left( 1-\mathcal{B}_{i}^{\left( \ell \right)
			}-B_{i}^{\left( \ell \right) }v\right) }\frac{\Gamma \left( 1+s\right) }{%
			\Gamma \left( 1-s\right) }\left( \frac{\phi _{\ell }\tau _{\varkappa }}{%
			\varrho _{\varkappa }}\right) ^{-v}\left( \frac{\delta _{n}\omega
			_{\varkappa }}{\varrho _{\varkappa }}\right) ^{-s}dvds,  \label{Integral_n}
		\end{eqnarray}
		\textcolor{black}{which concludes the proof.}
	\end{flushleft}
	
	\section*{Appendix D: Proof of corollary \protect\ref{corolidealasep}}
	
	\begin{flushleft}
		By setting $\omega _{\varkappa }=\tau _{\varkappa }=1,\left( \ref{Integral_n}%
		\right) $ can be rewritten using the change of variable $w=-s-v$ as
		\begin{eqnarray}
		\mathcal{I}_{n}^{\left( \varkappa \right) } &=&1-\frac{1}{2\pi j}\sum_{\ell
			=1}^{L}\frac{\psi _{\ell }}{\phi _{\ell }}\int_{\mathcal{L}_{v}^{\left( \ell
				\right) }}\frac{\Gamma \left( v\right) \prod\limits_{i=1}^{m}\Gamma \left(
			\mathcal{B}_{i}^{\left( \ell \right) }+B_{i}^{\left( \ell \right) }v\right)
			\prod\limits_{i=1}^{n}\Gamma \left( 1-\mathcal{A}_{i}^{\left( \ell \right)
			}-A_{i}^{\left( \ell \right) }v\right) }{\prod\limits_{i=n+1}^{p}\Gamma
			\left( \mathcal{A}_{i}^{\left( \ell \right) }+A_{i}^{\left( \ell \right)
			}v\right) \prod\limits_{i=m+1}^{q}\Gamma \left( 1-\mathcal{B}_{i}^{\left(
				\ell \right) }-B_{i}^{\left( \ell \right) }v\right) }  \notag \\
		&&\times \left( \frac{\phi _{\ell }}{\delta _{n}}\right) ^{-v}\underbrace{%
			\frac{1}{2\pi j}\int_{\mathcal{L}_{w}^{\left( \ell \right) }}\Gamma \left(
			w\right) \frac{\Gamma \left( 1-w-v\right) }{\Gamma \left( 1+w+v\right) }%
			\left( \frac{\varrho _{\varkappa }}{\delta _{n}}\right) ^{w}dw}_{\mathcal{M}%
			_{\varkappa }}dv. \label{Integrall-n}
		\end{eqnarray}%
		\textcolor{black}{Again, for significantly smaller values of $\varrho _{\varkappa }$ (i.e., $\varrho _{\varkappa }<<1$) and similarly to $\mathcal{K}_{\varkappa }$, the inner integral $\mathcal{M}_{\varkappa }$ can
			be expressed as}
		\begin{equation}
		\mathcal{M}_{\varkappa }\mathcal{=}\sum_{l=0}^{\infty }\frac{\left(
			-1\right) ^{l}\Gamma \left( 1+l-v\right) }{l!\Gamma \left( 1-l+v\right) }%
		\left( \frac{\varrho _{\varkappa }}{\delta _{n}}\right) ^{l}.
		\label{M-residu}
		\end{equation}%
		Specifically, for $\varrho _{\varkappa }=0$, $\left( \ref{M-residu}\right) $
		is reduced to
		\begin{equation}
		\mathcal{M}_{\varkappa }\mathcal{=}\frac{\Gamma \left( 1-v\right) }{\Gamma
			\left( 1+v\right) }.  \label{M-res2}
		\end{equation}%
		\textcolor{black}{As a result, $\left( \ref{Integrall-n}\right)$ can be
			reexpressed by incorporating $\left( \ref{M-res2}\right)$ as}
		\begin{eqnarray}
		\mathcal{I}_{n}^{\left( \varkappa \right) } &=&1-\frac{1}{2\pi j}\sum_{\ell
			=1}^{L}\frac{\psi _{\ell }}{\phi _{\ell }}\int_{\mathcal{L}_{v}^{\left( \ell
				\right) }}\frac{\Gamma \left( v\right) \prod\limits_{i=1}^{m}\Gamma \left(
			\mathcal{B}_{i}^{\left( \ell \right) }+B_{i}^{\left( \ell \right) }v\right)
		}{\Gamma \left( 1+v\right) \prod\limits_{i=n+1}^{p}\Gamma \left( \mathcal{A}%
			_{i}^{\left( \ell \right) }+A_{i}^{\left( \ell \right) }v\right) }  \notag \\
		&&\times \frac{\Gamma \left( 1-v\right) \prod\limits_{i=1}^{n}\Gamma \left(
			1-\mathcal{A}_{i}^{\left( \ell \right) }-A_{i}^{\left( \ell \right)
			}v\right) }{\prod\limits_{i=m+1}^{q}\Gamma \left( 1-\mathcal{B}_{i}^{\left(
				\ell \right) }-B_{i}^{\left( \ell \right) }v\right) }\left( \frac{\phi
			_{\ell }}{\delta _{n}}\right) ^{-v}dv.
		\end{eqnarray}%
		\textcolor{black}{Thus, the proof is concluded.}
	\end{flushleft}
	
	\section*{Appendix E: Proof of proposition \protect\ref{propasympAsep}}
	
	\begin{flushleft}
		By plugging $\left( \ref{ApproFHF}\right) $ in $\left( \ref{EqIntegral-n}%
		\right) $ and using \cite[Eq. (2.9.4)]{kilbas}, yields
		\textcolor{black}{the
			following approximate expression in the high SNR regime}
		\begin{eqnarray}
		\mathcal{I}_{n}^{\left( \varkappa \right) } &\sim &1-\sum_{\ell =1}^{L}\frac{%
			\psi _{\ell }}{\phi _{\ell }}E_{\ell }\left[ 1-\exp \left( -\frac{\delta
			_{n}\omega _{\varkappa }}{\varrho _{\varkappa }}\right) \right] +\delta
		_{n}\sum_{\ell =1}^{L}\frac{\psi _{\ell }}{\phi _{\ell }}\sum_{i=1}^{m}%
		\mathcal{F}_{\ell ,i}\left( \frac{\phi _{\ell }\tau _{\varkappa }}{\omega
			_{\varkappa }}\right) ^{\frac{\mathcal{B}_{i}^{\left( \ell \right) }}{%
				B_{i}^{\left( \ell \right) }}}  \notag \\
		&&\times \frac{1}{2\pi j}\int_{\mathcal{L}_{s}^{\left( \ell \right) }}\Gamma
		\left( s\right) \left( \delta _{n}\right) ^{-s}\int_{0}^{\frac{\omega
				_{\varkappa }}{\varrho _{\varkappa }}}\gamma ^{\frac{\mathcal{B}_{i}^{\left(
					\ell \right) }}{B_{i}^{\left( \ell \right) }}-s}\left( 1-\frac{\varrho
			_{\varkappa }}{\omega _{\varkappa }}\gamma \right) ^{\frac{\mathcal{B}%
				_{i}^{\left( \ell \right) }}{B_{i}^{\left( \ell \right) }}}d\gamma ds.\label{HH}
		\end{eqnarray}
		Now, using $\left( \ref{Hffox}\right) $, \cite[Eq. (2.9.15)]{kilbas}, and
		\cite[Eq. (1.3.194)]{integraltable}, one obtains
		\begin{eqnarray}
		\mathcal{I}_{n}^{\left( \varkappa \right) } &\sim &\exp \left( -\frac{\delta
			_{n}\omega _{\varkappa }}{\varrho _{\varkappa }}\right) +\delta
		_{n}\sum_{\ell =1}^{L}\frac{\psi _{\ell }\omega _{\varkappa }}{\phi _{\ell
			}\varrho _{\varkappa }}\sum_{i=1}^{m}\frac{\mathcal{F}_{\ell ,i}\left( \frac{%
				\phi _{\ell }\tau _{\varkappa }}{\omega _{\varkappa }}\right) ^{\frac{%
					\mathcal{B}_{i}^{\left( \ell \right) }}{B_{i}^{\left( \ell \right) }}}}{%
			\Gamma \left( \frac{\mathcal{B}_{i}}{B_{i}}\right) }\left( \frac{1}{2\pi j}%
		\right) ^{2}\int_{\mathcal{L}_{s}^{\left( \ell \right) }}\int_{\mathcal{L}%
			_{v}^{\left( \ell \right) }}\Gamma \left( s\right) \Gamma \left( v\right)
		\notag \\
		&&\times \frac{\Gamma \left( 1+\frac{\mathcal{B}_{i}^{\left( \ell \right) }}{%
				B_{i}^{\left( \ell \right) }}-s-v\right) }{\Gamma \left( 2+\frac{\mathcal{B}%
				_{i}^{\left( \ell \right) }}{B_{i}^{\left( \ell \right) }}-s-v\right) }%
		\Gamma \left( \frac{\mathcal{B}_{i}^{\left( \ell \right) }}{B_{i}^{\left(
				\ell \right) }}-v\right) \left( \frac{\delta _{n}\omega _{\varkappa }}{%
			\varrho _{\varkappa }}\right) ^{-s}\left( -1\right) ^{-v}dsdv.\label{RRR}
		\end{eqnarray}%
		Finally, substituting (\ref{RRR}) into (\ref{Psaprox}), yields (\ref%
		{AsymptoASEP}), which concludes the proof.
	\end{flushleft}
	
\end{document}